\documentclass[10pt, conference, letterpaper]{IEEEtran}

\usepackage{cite}
\ifCLASSINFOpdf
  \usepackage[pdftex]{graphicx}
\else
  \usepackage[dvips]{graphicx}
\fi
\usepackage[paper=letterpaper,left=0.64in,right=0.64in,top=0.75in,bottom=1.05in]{geometry}

\let\labelindent\relax
\usepackage{enumitem}
\setlist[enumerate]{leftmargin= 0.5 cm}
\setlist[itemize]{leftmargin=0.3 cm}

\usepackage{epstopdf}
\usepackage{graphicx}
\usepackage{subfigure}
\usepackage{amsmath}
\usepackage{amsthm}
\usepackage{algorithm}
\usepackage{algorithmic}
\usepackage{amssymb}
\usepackage[mathscr]{eucal}
\usepackage{url}
\usepackage{thmtools}
\usepackage{color}

\usepackage{multirow}

\usepackage{lipsum}
\usepackage{capt-of}

\makeatletter
\newcommand\subparagraph{%
  \@startsection{subparagraph}{5}
  {\parindent}
  {3.25ex \@plus 1ex \@minus .2ex}
  {-1em}
  {\normalfont\normalsize\bfseries}}
\makeatother
\usepackage{titlesec} 
\let\subparagraph\relax

\theoremstyle{definition}
\newtheorem{defi}{Definition}
\theoremstyle{theorem}
\newtheorem{lemma}{Lemma}
\newtheorem{theo}{Theorem}

\declaretheorem[style=definition,name=Example,qed=$\blacksquare$]{Example}

\DeclareMathAlphabet{\mathpzc}{OT1}{pzc}{m}{it} 

\makeatletter
\newcounter{procedure}
\newfloat{procedure}{htbp}{loa}[procedure]
\floatname{procedure}{Procedure}

\makeatother



\renewcommand{\baselinestretch}{0.91} 

\begin{document}
\title{Service Overlay Forest Embedding for \\Software-Defined Cloud Networks}
\author{Jian-Jhih Kuo\IEEEauthorrefmark{1}, Shan-Hsiang Shen\IEEEauthorrefmark{1}\IEEEauthorrefmark{2}, Ming-Hong Yang\IEEEauthorrefmark{1}\IEEEauthorrefmark{3}, De-Nian Yang\IEEEauthorrefmark{1}, Ming-Jer Tsai\IEEEauthorrefmark{4} and Wen-Tsuen Chen\IEEEauthorrefmark{1}\IEEEauthorrefmark{4}\\ 
\IEEEauthorrefmark{1}Inst. of Information Science, Academia Sinica, Taipei, Taiwan\\
\IEEEauthorrefmark{2}Dept. of Computer Science \& Information Engineering,\\ National Taiwan University of Science \& Technology, Taipei, Taiwan\\
\IEEEauthorrefmark{3}Dept. of Computer Science \& Engineering, University of Minnesota, Minneapolis MN, USA\\
\IEEEauthorrefmark{4}Dept. of Computer Science, National Tsing Hua University, Hsinchu, Taiwan\\
E-mail: \{lajacky,sshen3,curtisyang,dnyang,chenwt\}@iis.sinica.edu.tw and mjtsai@cs.nthu.edu.tw}
\maketitle

\begin{abstract}
Network Function Virtualization (NFV) on Software-Defined Networks (SDN) can effectively optimize the allocation of Virtual Network Functions (VNFs) and the routing of network flows simultaneously. Nevertheless, most previous studies on NFV focus on unicast service chains and thereby are not scalable to support a large number of destinations in multicast. On the other hand, the allocation of VNFs has not been supported in the current SDN multicast routing algorithms. In this paper, therefore, we make the first attempt to tackle a new challenging problem for finding a service forest with multiple service trees, where each tree contains multiple VNFs required by each destination. Specifically, we formulate a new optimization, named Service Overlay Forest (SOF), to minimize the total cost of all allocated VNFs and all multicast trees in the forest. We design a new $3\rho_{ST}$-approximation algorithm to solve the problem, where $\rho _{ST}$ denotes the best approximation ratio of the Steiner Tree problem, and the distributed implementation of the algorithm is also presented. Simulation results on real networks for data centers manifest that the proposed algorithm outperforms the existing ones by over 25\%. Moreover, the implementation of an experimental SDN with HP OpenFlow switches indicates that SOF can significantly improve the QoE of the Youtube service.
\end{abstract}

\section{Introduction}
\label{sec:intro} 
The media industry is now experiencing a major change that alters user subscription patterns and thereby inspires the architects to rethink the design \cite{FutureMediaServiceArchi}. For example, the live video streaming on Anvato\cite{Anvato} enables online video editing for content providers, ad insertion for advertisers, caching, and transcoding for heterogeneous user devices.
Google has acquired Anvato with the above abundant functions and integrated its architecture into Google Cloud to develop the next-generation Youtube.\footnote{http://www.msn.com/en-us/news/technology/google-buys-a-backbone-for-pay-tv-services/ar-BBu61eB?li=AA4Zoy\&ocid=spartanntp} Therefore, it is envisaged that the next-generation Youtube requires more computation functionalities and resources in the cloud. 
For distributed collaborative virtual reality (VR), it is also crucial to allocate distributed computation resources for important tasks such as collision detection, geometric constraint matching, synchronization, view consistency, concurrency and interest management \cite{VR1, VR2, VR3}. 
Network Function Virtualization (NFV) has been regarded as a promising way \cite{FutureMediaServiceArchi,needForNewBroadcastProduction} that exploits Virtual Machines (VMs) to divide the required function into building blocks connected with a service chain \cite{IntroduceNFV}. A service chain passes through a set of Virtual Network Functions (VNFs) in sequence, and Netflix \cite{Netflix} has adopted AWS \cite{AWS} to support the service chains. Current commercial solutions usually assign an individual service chain for each end user for unicast \cite{SurveyCurrentNfvImplementation,ServiceChain,NFVchainInOpticalDC-JLT,ServiceChaining}. Nevertheless, it is expected that this approach is not scalable because duplicated VNFs and network traffic are involved to serve all users if they require the same content, such as live/linear content broadcast. 
The global consumer research \cite{TVtrend2016} manifests that although the unicast video on demand becomes more and more popular, the live/linear content broadcast and multicast nowadays still account for over 50\% of viewing hours per week, from companies such as Sony Crackle \cite{Crackle} and Pluto TV \cite{PlutoTV}, because it effectively attracts the users through a shared social experience to instantly access the contents. However, currently there is no effective solution to support large-scale content distributions with abundant computation functionalities for content providers and end users.

For scalable one-to-many communications, multicast exploits a tree to replicate the packets in branching routers. Compared with unicast flows, a multicast tree can effectively reduce the bandwidth consumption in backbone networks by over 50\% \cite{BenefitOfMulticast}, especially for multimedia traffic \cite{MulticastAdvantage}. Currently, shortest-path trees are employed by Internet standards (such as PIM-SM \cite{PIM-SM}) because they can be efficiently constructed in a distributed manner. Nevertheless, the routing is not flexible since the path from the source to each destination needs to follow the corresponding shortest path. Recently, the flexible routing for traffic engineering becomes increasingly important with the emergence of Software-Defined Networks (SDNs), whereas centralized computation can be facilitated in an SDN controller to find the optimal routing, such as Steiner Tree \cite{SteinerTreeBestRatio} in Graph Theory or its variations  \cite{reliableMulticastforSDN,multicastTEforSDN}. Thus, multicast traffic engineering has been regarded very promising for SDNs.

Nevertheless, the above approaches and other existing multicast routing algorithms \cite{IntelligentSchedulingMulticastVoD,BestEffortPatchingMulticastTrueVoD} are not designed to support NFV because the nodes (e.g., the source and destinations) that need to be connected in a tree are specified as the problem input. On the contrary, here VNFs are also required to be spanned in a tree for NFV, and the problem is more challenging since VMs also need to be selected, instead of being assigned as the problem input. Moreover, multicast NFV indeed is more complicated when it is necessary to employ multiple multicast trees as a forest for a group of destinations, and this feature is crucial for Content Deliver Networks (CDNs) with multiple video source servers. In this case, the video source also needs to be chosen for each end user \cite{hofmann2005content}.

In this paper, therefore, we make the first attempt to explore the resource allocation problem (i.e., both the VM selection, source selection, and the tree routing) for a \textit{service forest} involving multiple multicast trees, where the path from a source to each destination needs to traverse a sequence of demanded services (i.e., a service chain) in the tree.\footnote{In this paper, we first consider the static multicast, and then clarify the static case is already a good step forward and discuss how to adapt the proposed algorithm to the dynamic case in Sections \ref{subsec: static case} and \ref{subsec: dynamic adjustment}, respectively.}
We formulate a new optimization problem for multi-tree NFV in software defined cloud
networks, named, Service Overlay Forest (SOF), to minimize the total cost of the selected VMs and trees. Given the sources and the destinations with each destination requiring a chain of services, the SOF problem aims at finding an overlay forest that 1) connects each destination to a source and 2) visits the demanded services in selected VMs in sequence before reaching the destinations.

Fig. \ref{fig:example 0} first compares a service tree and a service forest. Fig. \ref{fig: input 0} is the input network with the
cost labeled beside each node and edge to represent the link connection cost
and the VM setup cost, respectively. Assume that there are two destinations 9 and 10, and their
demanded service chain consists of two VNFs, $f_{1}$ and $f_{2}$ in order. 
A Steiner tree in Fig. \ref{fig: ST 0} spanning source node 1 and both destinations incurs the total cost as 34 if VMs 2 and 3 are employed. Note that the edge between VMs 2 and 3 is visited twice to reach destination 10, and the cost of the edge is thus required to be included twice. More specifically, the edge costs from source 1 to VM 3 ($f_1$), from VM 3 to VM 2 ($f_2$), and from VM 2 to destinations 9, 10, are $1$, $3$, $20+1+1+1+3+1+1=28$, respectively. Thus, the edge cost is $1+3+28=32$ while the node cost is $1+1=2$.
By contrast, the cost of a service forest with two trees and four VMs selected is 14 in Fig. \ref{fig: SOF 0}, which significantly reduces the cost by about 60 \%. This example
manifests that consolidating the services in few VMs may not always lead to
the smallest cost because the edges to connect multiple destinations are
also important. Therefore, multiple trees with multiple sources are promising
to further reduce the cost.\footnote{In this paper, we assume that the setup cost for a source node is negligible. The source with the setup cost is further discussed in Appendix \ref{sec: source cost}.} 

In this paper, we first prove that the problem is NP-hard. To investigate the
problem in depth, we will step-by-step reveal the thinking process of the
algorithm design from the single-source case to the general case, and then propose a $3\rho _{ST}$-approximation algorithm,\footnote{Compared with the traditional Steiner Tree problem, the problem considered in this paper is more difficult due to new  SDN/NFV constraints involved. Indeed, several recent research works \cite{reliableMulticastforSDN,multicastTEforSDN} on SDN multicast and NFV service chain embedding (e.g., \cite{ServiceChain,ServiceChaining}) have massive approximation ratios (e.g. $O(|D|)$, where $|D|$ denotes the number of destinations, and $O(|\mathcal{C}|)$, where $|\mathcal{C}|$ denotes the length of demanded service chain). By contrast, the approximation ratio of this paper is $3\rho_{ST}$, where $\rho_{ST}$ is the best approximation ratio of the Steiner Tree problem (e.g., the current best one is 1.39), which is smaller than the above works. Moreover, the simulation results manifest that empirically the performance is very close to the optimal solutions obtained by the proposed Integer Programming formulation.}
named Service Overlay Forest Deployment Algorithm (SOFDA) for the general case, where $\rho _{ST}$ denotes the best approximation ratio of the Steiner Tree problem (e.g., the current best one is 1.39). The single-source case is more difficult than the traditional Steiner tree problem because not only the terminal nodes (i.e., source and destinations) need to be spanned but also a set of VMs is required to be selected and spanned to install VNFs in sequence. Also, the general case is more challenging than the single-source case, because a service tree is necessary to be created for each source, and the \emph{VNF conflict} (i.e., a VM is allocated with too many VNFs from multiple trees) tends to occur in this case. 

Therefore, SOFDA is designed to 1) assign multiple sources for varied trees with multiple VMs, 2) allocate the VMs for each tree to provide a service chain for each destination, and 3) find the routing of each tree to span the selected source, VMs, and destinations. Simulation on real topologies manifests that SOFDA can effectively reduce the total cost for data center networks. In addition, a distributed SOFDA is proposed to support the multi-controller SDNs. Implementation on an experimental SDN for Youtube traffic also indicates that the user QoE can be significantly improved for transcoded and watermarked video streams. The rest of this paper is organized as
follows. The related works are summarized in Section \ref{sec: related work}. We formulate the SOF problem in Section \ref{sec: SOF problem} and design the approximation algorithms in Sections \ref{sec: one source} and \ref{sec: general case}. The distributed algorithm is presented in Section \ref{sec: decentralize}. Some important issues are discussed in Section \ref{sec: discussion}. 
The simulation and implementation results are presented in Section \ref{sec: simulation}, and we conclude this paper in Section \ref{sec: conclusion}.

\begin{figure}[t]
\centering
\subfigure[]{\includegraphics[width=2.9cm]{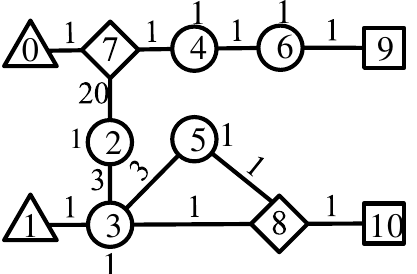}\label{fig: input
0}}  \hfill  \subfigure[]{\includegraphics[width=2.9cm]{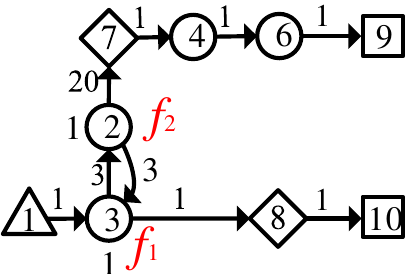}%
\label{fig: ST 0}}  \hfill  \subfigure[]{%
\includegraphics[width=2.9cm]{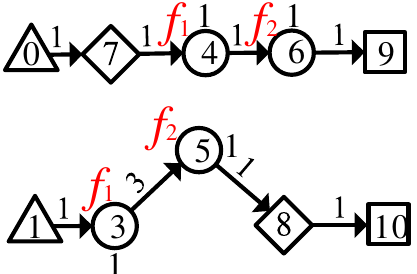}\label{fig: SOF 0}}\\[-2pt]
\subfigure{\includegraphics[width=5.2cm]{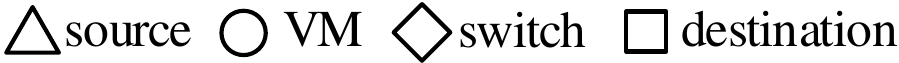}\label{fig:
legend 0}} 
\caption{Comparison of service trees and service forests. (a) Input network. 
(b) Steiner tree with predetermined VMs. (c) Service Overlay Forest.}
\label{fig:example 0}
\end{figure}


\section{Related Work}\label{sec: related work}
Traffic engineering for unicast service chains in SDN has drawn increasing
attention recently. Lukovszki \textit{et al.} \cite{ServiceChain} point out
that the length of a service chain is necessary to be bounded and present an
efficient online algorithm to maximize the number of deployed service
chains, whereas the maximal number of VMs hosted on a node is also
guaranteed. Xia \textit{et al.} \cite{NFVchainInOpticalDC-JLT} jointly
consider the optical and electrical domains and minimize the number of
domain conversions in all service chains. Moreover, Kuo \textit{et al.} \cite%
{ServiceChaining} strike the balance between link utilization and server
usage to maximize the total benefit. Nevertheless, the above studies only
explore unicast routing for service chains and do not support multicast.

Multicast traffic engineering for SDN is more complicated than traditional
unicast traffic engineering. Huang \textit{et al.} \cite%
{ScalableMulticastforSDN} first incorporate the flow table scalability in
the design of the multicast tree routing in SDN. Shen \textit{et al.} \cite%
{reliableMulticastforSDN} then further consider the packet loss recovery in
reliable multicast routing for SDN. Recently, the routing of multiple trees in
SDN \cite{multicastTEforSDN} has been studied to ensure that the routing follows
both the link capacity and the TCAM size. The problem is more challenging
due to the above two constraints, and the best approximation ratio that can
be achieved is only $D$ (i.e., the maximum number of destinations in a
tree). However, the dimension of service allocation in VMs has not been
explored in the above work. Recently, special cases on a tree \cite%
{NFV-multicast,NFV-Tree} with only one source and one VM have been explored. 
Overall, the above approaches are not designed to support a service forest with multiple VNFs and multiple trees, and the problem here is more challenging because \emph{VNF conflict} due to the overlapping of trees will occur. To the best knowledge of the authors, this paper is the first
one that explores both routing and VM selection for multiple trees to
construct a forest in SDN. As explained in Section \ref{sec:intro}, the
service forest is important for many emerging and crucial multimedia
applications in CDN that require intensive cloud computing.


\section{The Service Overlay Forest Problem} \label{sec: SOF problem}
A service overlay forest consists of a set of service overlay trees. Each service overlay tree spans one source, a set of VMs for enabled VNFs, and a subset of destinations. Any two service overlay trees do not overlap since each destination only needs to connect to a source via a service chain in a tree. In the following, we first formally define the problem. We are given:
\begin{enumerate}
\item a network $G=\{V= M \cup U,E\}$, where each link $e \in E$ is associated with a nonnegative cost $c(e)$ denoting the connection cost of link $e$ to forward the demand of destinations, each virtual machine (VM) $v \in M$ is associated with a nonnegative cost $c(v)$ denoting the setup cost of VM $v$ to run a virtual network function (VNF), and each switch $v \in U$ is associated with cost $0$,
\item a set of destinations $D\subseteq V$ requesting the same demand,
\item a set of sources $S\subseteq V$ having the demands of destinations, and
\item a chain of VNFs $\mathcal{C}=(f_1,f_2,\cdots, f_{|\mathcal{C}|})$ required to process the demand of destinations. 
\end{enumerate}

The \textbf{S}ervice \textbf{O}verlay \textbf{F}orest (SOF) problem is to construct a service overlay forest consisting of the service overlay trees with the roots in $S$, the leaves in $D$, and the remaining nodes in $V$, so that there exists a chain of VNFs from a source to each destination. A chain of VNFs is represented by a \textit{walk}, which is allowed to traverse a node (i.e., a VM or a switch) multiple times. In each walk, a \textit{clone} of a node and the corresponding incident links are created to foster an additional one-time pass of the node, and only 
one of its clones is allowed to run VNF to avoid duplicated counting of the setup cost. For example, in the second feasible forest (colored with light gray) of Fig. \ref{fig: feasible solution 1}, a walk from source 1 to destination 8 passes VM 2 twice without running any VNF, and there are two clones of VM 2 on the walk. For each destination $t \in D$, SOF needs to ensure that there exists a path with clone nodes (i.e., a walk on the original $G$) on which $f_1,f_2,\cdots, f_{|\mathcal{C}|}$ are performed in sequence from a source $s \in S$ to $t$ in the service overlay forest.

The objective of SOF is to minimize the total setup and connection cost of the service overlay forest, where the setup and connection costs denote the total cost of the VMs and links, respectively. Note that the cost of a link in $G$ is counted twice if the link is duplicated because its terminal nodes are cloned. In this paper, it is assumed that a VM can run at most one VNF in the network $G$. The scenario that requires a VM to support multiple VNFs can be simply addressed by first replicating the VM multiple times in the input graph $G$. 

\begin{figure}
	\center
	\subfigure[]{\includegraphics[width=4.3cm]{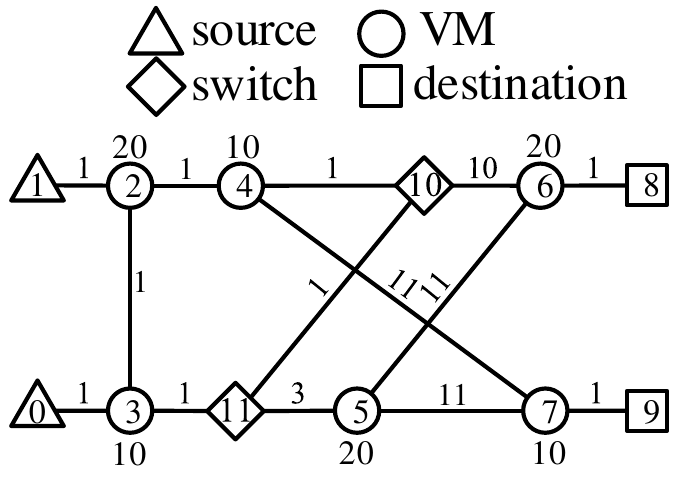}\label{fig: given}}\hfil
	\subfigure[]{\includegraphics[width=4.3cm]{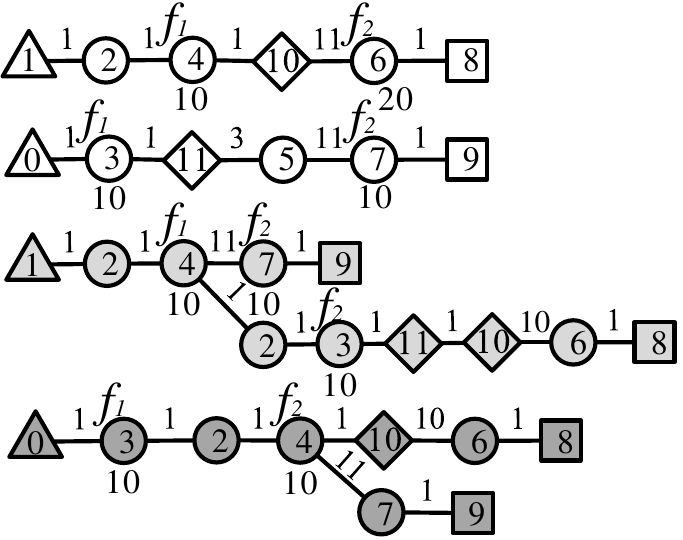}\label{fig: feasible solution 1}}\hfil
	\caption{Example of service overlay forests. (a) The input network $G$. (b) The service overlay forests with $\mathcal{C}=(f_1,f_2)$ constructed for $G$.} \label{fig: example1}
\end{figure}
\begin{Example}
Fig. \ref{fig: example1} presents three examples for the service overlay forests. The first service overlay forest consists of two service overlay trees, where the demand of destination $8$ (or $9$) is routed from source $1$ (or $0$) along the walk $(1,2, 4, 10, 6, 8)$ (or $0, 3, 11, 5, 7, 9$), and the demand is processed by VNFs $f_1$ and $f_2$ at VMs $4$ and $6$ (or $3$ and $7$), respectively. The total cost of the first service overlay forest is $82$, where the setup cost and connection cost are $50$ and $32$, respectively. In the second service overlay forest (including only one tree), source 1 first routes the demand to VM 4 for VNF $f_1$. Subsequently, VM 4 forwards the demand to VM 7 and VM 2 for VNF $f_2$, respectively. Finally, the demand is forwarded towards destinations $8$ and $9$, respectively. The setup cost and connection costs of the second service overlay forest are $30$ and $29$, respectively. In the third service overlay forest (tree), the demand is first routed from source $1$ to VM $3$ for VNF $f_1$ and then toward VM $4$ for VNF $f_2$, and finally to destinations $8$ and $9$, respectively. The third service overlay forest is an optimal service overlay forest with the setup cost and connection cost as $20$ and $27$, respectively.
\end{Example}

\subsection{Integer Programming} 
\label{subsec: IP}
In the following, we present the Integer Programming (IP) formulation for
SOF. Our IP formulation first identifies the service chain 
for each destination and then constructs the whole service forest accordingly. To find
the walk of the service chain, we first assign the VMs corresponding to
each VNF in the walk and then find the routing of the walk between every
two consecutive VMs. More specifically, SOF includes the
following binary decision variables. Let $\gamma _{d,f,u}$ denote if node $u$
is assigned as the enabled VM for VNF $f$ in the walk to
destination $d$. Let $\pi _{d,f,u,v}$ denote if edge $e_{u,v}$ is located in
the walk connecting the enabled VM of VNF $f$ and the enabled
VM of the next VNF $f_{N}$. 
Note that the above walk will
belong to a service tree rooted at the enabled candidate node of VNF $f$. 
Therefore, to find the service forest for $f$, let binary variable $\tau
_{f,u,v}$ represent if edge $e_{u,v}$ is located in the forest. On the other
hand, binary variable $\sigma _{f,u}$ represents if node $u$ is assigned as
the enabled VM of service $f$ for the whole service forest.
Notice that each destination $d$ may desire a different VNF $f$ on the
same enabled VM $u$ according to $\gamma _{d,f,u}$, but the
constraint later in this section will ensure that only one VNF is
allowed to be allocated to $u$ by properly assigning $\sigma _{f,u}$
accordingly.

The objective function for SOF is as follows.%
\begin{equation*}
\min \sum \limits_{f\in \mathcal{C}}\sum \limits_{u\in V}c(u)\sigma _{f,u}+\min
\sum \limits_{f\in \mathcal{C}}\sum \limits_{e_{u,v}\in E}c(e_{u,v})\tau _{f,u,v},
\end{equation*}%
where the first term represents the total setup cost of all VMs, and the
second term is the connection cost of the service forest. The IP formulation
contains the following constraints.

1) Service Chain 
Constraint. The following four constraints first assign the
enabled VM for each service chain. 

\begin{center}
\begin{tabular}{cc}
$\sum \limits_{s\in S}\gamma _{d,f_{S},s}=1$, $\forall d\in D,$ & (1) \\ 
\vspace{2pt} $\sum \limits_{u\in M}\gamma _{d,f,u}=1$, $\forall d\in D,f\in \mathcal{C},
$ & (2) \\ 
$\gamma _{d,f_{D},d}=1$, $\forall d\in D,$ & (3) \\ 
$\gamma _{d,f_{D},u}=0$, $\forall d\in D,u\in V-\{d\}.$ & (4)%
\end{tabular}
\end{center}

Constraint (1) ensures that each destination chooses one source $s$ in $S$
as its service source, where $f_{S}$ denotes the function as the source of
the service chain. 
Constraint (2) finds a node $u$ from $M$ as the enabled
VM of each VNF $f$ for each destination. Constraints (3) and (4) assign only
destination $d$ for function $f_{D}$, where $f_{D}$ denotes the function as
the destination of the service chain. 
Here notations $f_{S}$ and $f_{D}$ are
incorporated in our IP formulation in order to support the routing
constraints described later. In other words, a service chain 
traverses the
nodes with $f_{S},f_{1},...,f_{\left \vert \mathcal{C}\right \vert },f_{D}$ sequentially.

2) Service Forest Constraint. The following two constraints assign the
enabled VM for the whole service forest.

\begin{center}
\begin{tabular}{cc}
\vspace{2pt} $\gamma _{d,f,u}\leq \sigma _{f,u}$, $\forall d\in D,f\in
\mathcal{C},u\in V,$ & (5) \\ 
\vspace{2pt} $\sum \limits_{f\in \mathcal{C}}\sigma _{f,u}\leq 1$, $\forall u\in V,$ & 
(6)%
\end{tabular}
\end{center}

Constraint (5) assigns $u$ as the enabled VM of VNF $f$ for
the whole service forest if $u$ has been selected by at least one
destination $d$ for VNF $f$. Constraint (6) ensures that each node $u$
is in charge of at most one VNF.

3) Chain 
and Forest Routing Constraints. The following two constraints find
the routing of the whole service forest.

\begin{center}
\begin{tabular}{cc}
$\sum \limits_{v\in N_{u}}\pi _{d,f,u,v}-\sum \limits_{v\in N_{u}}\pi
_{d,f,v,u}\geq \gamma _{d,f,u}-\gamma _{d,f_{N},u}$, &  \\ 
$\forall d\in D,f\in \mathcal{C}\cup \{f_{S}\},u\in V,$ & (7) \\ 
$\pi _{d,f,u,v}\leq \tau _{f,u,v}$, $\forall d\in D,f\in \mathcal{C}\cup
\{f_{S}\},e_{u,v}\in E.$ & (8)%
\end{tabular}
\end{center}

Constraint (7) is the most complicated one. It first finds the routing of the service chain 
for each
destination $d$. For the source $u$ of a service chain, 
$\gamma _{d,f_{S},u}=1
$ and $\gamma _{d,f_{N},u}=\gamma _{d,f_{1},u}=0$, where $f_{1}$ is the
first VNF in $\mathcal{C}$. In this case, the constraint becomes 
\begin{equation*}
\sum \limits_{v\in N_{u}}\pi _{d,f_{S},u,v}-\sum \limits_{v\in N_{u}}\pi
_{d,f_{S},v,u}\geq 1.
\end{equation*}%
It ensures that at least one edge $e_{u,v}$ incident from $u$ is selected
for the service chain 
because no edge $e_{v,u}$ incident to $u$ is chosen
(i.e., $\sum \limits_{v\in N_{u}}\pi _{d,f_{S},v,u}=0$ for the source $u$).
By contrast, for any intermediate switch $u$ in the walk from the source to
the enabled VM of $f_{1}$, $\gamma _{d,f_{S},u}=0$ and $\gamma
_{d,f_{1},u}=0$, and the constraint becomes%
\begin{equation*}
\sum \limits_{v\in N_{u}}\pi _{d,f_{S},v,u}\leq \sum \limits_{v\in N_{u}}\pi
_{d,f_{S},u,v}.
\end{equation*}%
When any edge $e_{v,u}$ incident to $u$ has been chosen in the walk, the
above constraint states that at least one edge $e_{u,v}$ incident from $u$
must also be selected in order to construct the service chain 
iteratively.
The above induction starts from the source of the walk to the previous node
of the enabled VM of $f_{1}$. Afterward, for the enabled VM
$u$ of $f_{1}$, $\gamma _{d,f_{1},u}=1$ and $\gamma _{d,f_{S},u}=0$,
and the constraint becomes 
\begin{equation*}
\sum \limits_{v\in N_{u}}\pi _{d,f_{S},v,u}-\sum \limits_{v\in N_{u}}\pi
_{d,f_{S},u,v}\leq 1.
\end{equation*}%
Since $\pi _{d,f_{S},v,u}=1$ for only one edge $e_{v,u}$ in the walk
incident to $u$, the above constraint is identical to $\sum \limits_{v\in
N_{u}}\pi _{d,f_{S},u,v}\geq 0$. Therefore, $\pi _{d,f_{S},u,v}$ is allowed
to be $0$ for every edge $e_{u,v}$ to minimize the objective function,
implying that no data of $f_{S}$ will be sent from the enabled VM
of $f_{1}$. By contrast, $\pi _{d,f_{1},u,v}$ will be $1$ for one edge $%
e_{u,v}$ due to constraint (6), implying that the enabled VM of $%
f_{1}$ will deliver the data in one edge incident from $u$, and the above
induction repeats sequentially for every service $f$ in $\mathcal{C}$ until it reaches
the destination $d$. Finally, constraint (8) states that any edge $e_{u,v}$ is
in the service forest if it is 
in the service chain 
for at least one
destination $d$.

\subsection{The Hardness}
The SOF problem is NP-hard since a metric version of the Steiner Tree problem (see Definition \ref{defi: ST}) can be reduced to the SF problem in polynomial time. 
The complete proof is presented in Appendix \ref{sec: hardness}.

\begin{defi}
\cite{SteinerTreeBestRatio} Given a weighted graph $G=\{V,E\}$ with 
edge costs, a root $r\in V$ and a node set $U\subseteq V\setminus\{r\}$, a Steiner Tree is a minimum spanning tree that roots at $s$ and spans all the nodes in $U$, where $U\neq\emptyset$. \label{defi: ST}
\end{defi}
\begin{theo}
The SOF problem is NP-hard. \label{theo: NP-hard}
\end{theo} 



\section{Special Case with Single Tree} \label{sec: one source}
In this subsection, we propose a $(2+\rho _{ST})$-approximation algorithm,
named Service Overlay Forest Deployment Algorithm with Single Source
(SOFDA-SS) to explore the fundamental characteristics of the problem, and a more complicated algorithm for the general case with multiple sources will be presented in the next section. 
SOFDA-SS includes the following two phases. The first phase chooses the most suitable VM to install the last VNF (i.e., called last VM in the rest of this paper) and then finds a minimum-cost service chain\footnote{The next section will extend the service chain into a service tree with multiple last VMs.} between the source and the last VM. 
Afterward, the second phase finds a minimum-cost Steiner tree to span the VM and all the destinations. The selection of the last VM is crucial due to the following trade-offs. Choosing a VM closer to the source tends to generate a shorter service chain, but it may create a larger service tree if the last VM is distant from all destinations. Also, it is important to address the trade-off between the setup cost and connection cost, because a VM with a smaller setup cost will sometimes generate a larger tree.  
The pseudo code of SOFDA-SS is presented in Appendix \ref{sec: pseudo code} (see Algorithm \ref{Algo: SOF_one_candidate_source}) .

Therefore, to achieve the approximation ratio, it is necessary for SOFDA-SS to carefully examine every possible VM to derive a Steiner tree and evaluate the corresponding cost. 
For every VM $u$, to obtain a walk $W_{G}$ (i.e., service chain) from source $s$ to $u$ with $|\mathcal{C}|$ VMs (so that the VNFs $f_{1},f_{2},\cdots ,f_{|\mathcal{C}|}$
can be installed in sequence) in $G$, we first propose a graph transformation from $G$ to $\mathcal{G}$ and then find the $k$-stroll \cite{k-MST03} from $s$ to $u$ defined as follows.

\begin{defi}
Given a weighted graph $\mathcal{G}=\{\mathcal{V},\mathcal{E}\}$ and two
nodes $s$ and $u$ in $\mathcal{V}$, the $k$-stroll problem is to find the
shortest walk that visits at least $k$ distinct nodes (including $s$ and $u$%
) from $s$ to $u$ in $\mathcal{G}$. \label{defi: k-stroll}
\end{defi}


SOFDA-SS constructs an instance $\mathcal{G}=\{\mathcal{V},\mathcal{E}\}$ of the $k$-stroll problem from $G$ as follows. Let $\mathcal{V}$ consist of $s$ and all VMs in $G$ (i.e., $\mathcal{V}=M\cup \{s\}$). Let $\mathcal{E}$ contain all edges
between any two nodes in $\mathcal{V}$ (i.e., $\mathcal{G}$ is a complete graph). 
The cost of the edge between nodes $v_1$ and $v_2$ in $\mathcal{E}$ is defined as follows,
\vspace{-0.1cm}
\begin{equation}\notag
c(v_1,v_2)=\sum_{(a,b) \in P}c((a,b))+
\begin{cases}
\begin{array}{ll}
\frac{c(u)+c(v_2)}{2} & \text{if } v_1=s,\\
\frac{c(v_1)+c(u)}{2} & \text{else if } v_2=s,\\
\frac{c(v_1)+c(v_2)}{2} & \text{otherwise},
\end{array}
\end{cases}
\end{equation}
where $u$ and $P$ denote the last VM and the shortest path between nodes $v_1$ and $v_2$ in $G$, respectively. In other words, the cost of each shortest path in $G$ is first included in the cost of the corresponding edge in $\mathcal{E}$. Afterward, since the data always enter and leave the VM running an intermediate VNF ($\neq f_{|\mathcal{C}|}$), the setup cost of the VM is shared by the incoming and outgoing edges of the VM. Finally, the setup cost of last VM $u$ is shared by the outgoing edge of $s$ and the incoming edge of $u$. The edge costs of $\mathcal{G}$ are assigned in the above way to ensure that the shortest walk with $|\mathcal{C}|$ VMs in $G$ is identical to the shortest path with $|\mathcal{C}|+1$ nodes in $\mathcal{G}$.

Clearly, $\mathcal{G}$ can be constructed in polynomial time. Then, SOFDA-SS finds a $k$-stroll walk $W_{\mathcal{G}}^{\prime }$ that visits exactly $|\mathcal{C}|+1$ distinct
nodes from source $s$ to the last VM $u$ (i.e., $k=|\mathcal{C}|+1$) in $\mathcal{G}$. Then, SOFDA-SS finds the corresponding walk $W_{G}$ (i.e., a service chain from $s$ to $u$ in $G$) that visits exactly $|\mathcal{C}|$ distinct VMs in $G$ by concatenating each shortest path corresponding to a selected edge in $|\mathcal{C}|$, and each path connects two consecutive nodes, $u_{j}$ and $u_{j+1}$, on walk $W_{\mathcal{G}}$, where $1\leq j\leq |\mathcal{C}|$. Finally, the demanded VNFs $f_1,f_2,...,f_{|\mathcal{C}|}$ can be deployed in order on the walk with $|\mathcal{C}|$ VMs from $s$ to $u$. 

\begin{Example}
Fig. \ref{fig: example2} presents an illustrative example for SOFDA-SS. First, for VM $7$, the walk $W_G=(u_1,u_2,\cdots, u_{|\mathcal{C}|+1})$ with $u_1=1$ and $u_{|\mathcal{C}|+1}=7$ is obtained as follows. An instance $\mathcal{G}=\{\mathcal{V},\mathcal{E}\}$ of the $k$-stroll problem is first constructed with $s=1$, $M=\{2,3,4,5,6,7\}$, and $u=7$, where $\mathcal{V}$ is set to $\{1,2,3,4,5,6,7\}$, $\mathcal{E}$ is set as $\{(x,y)|x,y \in \mathcal{V}\}$, the cost of the edge between nodes $1$ and $6$ is set to $c((1,2))+c((2,4))+c((4,6))+\frac{c(5)+c(6)}{2}=14$, and the cost of the edge between nodes $2$ and $6$ is set as $c((2,4))+c((4,6))+\frac{c(2)+c(6)}{2}=13$. Subsequently, we acquire a walk $W'_\mathcal{G}=W_\mathcal{G}=(1,2,4,3,5,7)$ in $\mathcal{G}$ and the corresponding walk $W_G=(1,2,4,2,3,5,7)$ in $G$. After $W_G$ is obtained, the service overlay forest with the last VM (i.e., 7) is constructed, where the demand is first routed from source $1$ to VM $7$ along the walk $W_G=(1,2,4,2,3,5,7)$, and $f_1$, $f_2$, $f_3$,$f_4$, and $f_5$ is processed at VMs $2$, $4$, $3$, $5$, and $7$, respectively. After finding the Steiner tree rooted at VM 7, the demand is then routed to destination $8$ by traversing switches $4$ and $6$, and directly to destination $9$. The total cost in the end of the second phase is 45.
\end{Example}

\begin{figure}
	\center
	\subfigure[]{\includegraphics[width=2cm]{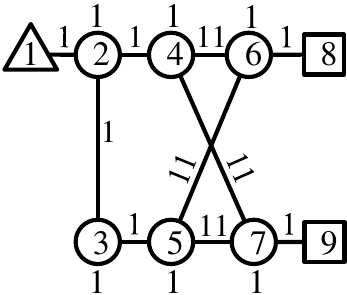}\label{fig: given}}\hfil
	\subfigure[]{\includegraphics[width=3.45cm]{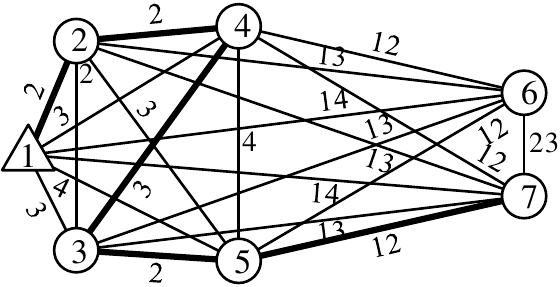}\label{fig: k-stroll instance}}\hfil
	\subfigure[]{\includegraphics[width=3.5cm]{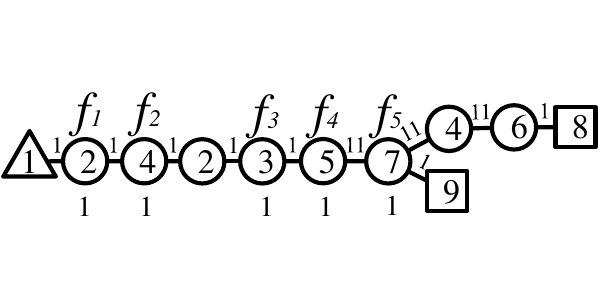}\label{fig: single chain solution}}
	\caption{Example of service overlay forest by SOFDA-SS. (a) The input network $G$. (b) The constructed instance of the k-Stroll problem $\mathcal{G}$, where the walk $W_\mathcal{G}$ between nodes 1 and 5 
is shown in bold. (c) The service overlay forest with $\mathcal{C}=(f_1,f_2,f_3,f_4,f_5)$ constructed for $G$.} \label{fig: example2}
\end{figure}

In the following, we present several important characteristics for graph $\mathcal{G}$, which play crucial roles to derive the approximation ratio. First, the cost of a walk $(u_{1},u_{2},\cdots ,u_{|\mathcal{C}|+1})$ from $s=u_{1}$ to the last VM $u=u_{|\mathcal{C}|+1}$ without traversing a node multiple times in $\mathcal{G}$ is equal to the sum of the total setup cost of $u_{2},u_{3},\cdots ,u_{k}$, plus the total connection cost of the shortest paths between every $u_{j}$ and $u_{j+1}$ for $j=1,2,\cdots ,k-1$ in $G$. 
Second, the edge costs in $\mathcal{G}$ satisfy the triangular inequality, as described in the following lemma. For readability, the detailed proof is presented in Appendix \ref{sec: triangular inequality}.

\begin{lemma} \label{lemma: triangular inequality}
The graph $\mathcal{G}$ satisfies triangular inequality.
\end{lemma}

Let $c(\mathcal{F}_{M}^{OPT})$ and $c(\mathcal{F}_{E}^{OPT})$ denote the setup and connection costs of the optimal service overlay forest $%
\mathcal{F}^{OPT}$, respectively. Based on the above two
characteristics, the following theorem derives the approximation ratio of SOFDA-SS.
The complete proof is presented in Appendix \ref{sec: SOFDA-SS ratio}.


\begin{theo} \label{theo: SOFDA-SS ratio}
	The cost of $F$ is bounded by $(2+\rho_{ST})c(\mathcal{F}^{OPT})$. That is, SOFDA-SS 
 is a $(2+\rho_{ST})$-approximation algorithm for the SOF problem with one tree. \label{theo: exist feasible}
\end{theo}

\textbf{Time Complexity Analysis.}
 SOFDA-SS constructs $|M|$ instances of the $k$-stroll problem, and each of them employs the Dijkstra algorithm $|M|$ times to compute the edge costs of each instance, where $O(T_d)$ denotes the time to run the Dijkstra algorithm. Moreover, let $O(T_k)$ denote the time to solve a $k$-stroll instance \cite{k-MST03}, and let $O(T_s)$ represent the time to append a Steiner tree by \cite{SteinerTreeBestRatio}. Therefore, the overall time complexity 
is $O(|M|(T_d|M|+T_k+T_s))$.

\section{General Case With Multiple Trees} \label{sec: general case}
In this section, we propose a $3\rho _{ST}$-approximation algorithm, named
Service Overlay Forest Deployment Algorithm (SOFDA), for the general SOF
problem with multiple sources. Different from SOFDA-SS, here we select
multiple sources to exploit multiple trees for further reducing the total
cost, and it is necessary to choose a different subset of destinations for each source to form a forest. In other words, both the last VMs and the set of destinations are necessary to be carefully chosen for the tree corresponding to each source. To effectively solve the above problem, our idea is to identify a short
service chain from each source to each destination as a candidate service
chain and then encourage more destinations to merge their service chains into a service tree,
and those destinations will belong to the same tree in this case. More
specifically, SOFDA first constructs an auxiliary graph $\mathbb{G}$ with each
candidate service chain represented by a new virtual edge connecting the
source and the last VM of the chain. Also, every source is connected to a
common virtual source. SOFDA finds a Steiner tree spanning the virtual source and all
destinations, and we will prove that the cost of the tree in $\mathbb{G}$ is no
greater than $3\rho_{ST}c(\mathcal{F}^{OPT})$.

Nevertheless, a new challenge arises here because the service chains corresponding to the selected virtual edges in the above approach may overlap in a few nodes in $G$, and the solution thereby is not feasible if any overlapping node in this case needs to support multiple VNFs (see the definition of SOF). SOFDA in Section \ref{subsec: service overlay tree} thereby revises the above
solution into multiple feasible trees, and we prove that SOFDA can still maintain the desired approximation in Section \ref{subsec: the general case}. 
The pseudo code of SOFDA is presented in Appendix \ref{sec: pseudo code} (see Algorithm \ref{Algo: SOF}).


\subsection{Cost-Bounded Steiner Tree}\label{subsec: the general case}
SOFDA first constructs an auxiliary graph $\mathbb{G}$ to effectively extract multiple service chains and group the destinations. Specifically, let $\mathbb{V_{S}}$ consist of the duplicate $\hat{v}$ of each source $v\in S$, and let $\mathbb{V_{M}}$ contain the duplicate $\hat{v}
$ of each VM $v\in M$. Therefore, $\mathbb{V}=V\cup \{\hat{s}\}\cup \mathbb{V_{S}}%
\cup \mathbb{V_{M}}$, where $\hat{s}$ denotes the \emph{virtual source}.
Also, let $\mathbb{E}_{\hat{s}\mathbb{S}}$ include 
the edges between $\hat{%
s}$ and $\hat{v}$ for each $\hat{v}\in \mathbb{V_{S}}$. Let $\mathbb{E}_{%
\mathbb{S}\mathbb{M}}$ consist of the \emph{virtual edges}
(representing the candidate service chain) between $\hat{v}$ and $\hat{u}$
for each $\hat{v}\in \mathbb{V_{S}}$ and $\hat{u}\in \mathbb{V_{M}}$, and
let $\mathbb{E}_{\mathbb{M}M}$ include the edges between $v$ and $\hat{v}$
for each $v\in M$. Then, $\mathbb{E}=E\cup \mathbb{E}_{\hat{s}\mathbb{S}%
}\cup \mathbb{E}_{\mathbb{S}\mathbb{M}}\cup \mathbb{E}_{\mathbb{M}M}$.
Moreover, the cost of each edge in $\mathbb{E}_{\hat{s}\mathbb{S}}\cup 
\mathbb{E}_{\mathbb{M}M}$ is assigned to 0, and the cost of the \emph{virtual edge} between $\hat{
v}\in \mathbb{V_{S}}$ and $\hat{u}\in \mathbb{V_{M}}$ in $\mathbb{E}_{%
\mathbb{S}\mathbb{M}}$ is equal to the cost of the $k$-stroll walk that visits $|\mathcal{C%
}|$ VMs between $v$ and $u$ in $G$. We first present an illustrative example for the above graph transformation. 

\begin{figure}
	\center
	\subfigure[]{\includegraphics[width=2cm]{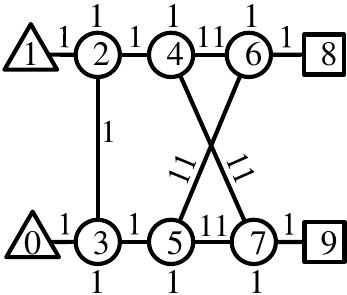}\label{fig: given}}\hfil
	\subfigure[]{\includegraphics[width=3.45cm]{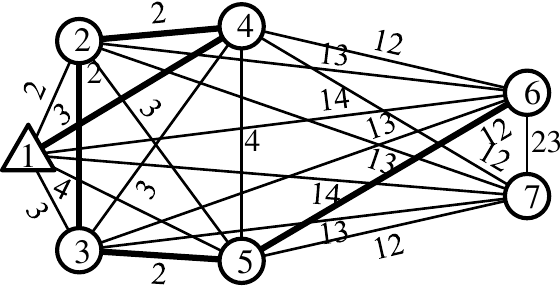}\label{fig: k-Stroll Instance C=5}}\hfil
	\subfigure[]{\includegraphics[width=3.5cm]{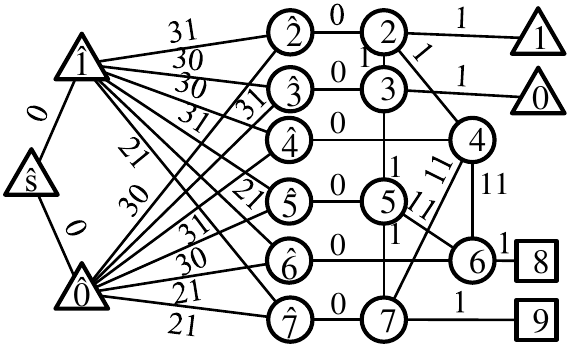}\label{fig: complete}}
	\caption{Example of instance construction of the Steiner tree problem. 
(a) The input network $G$. (b) The instance of the k-Stroll problem $\mathcal{G}$ constructed, 
where the walk $W_\mathcal{G}$ between nodes $1$ and $6$ 
 is shown in bold. (c) The instance of the Steiner tree problem.} \label{fig: example3}
\end{figure}

\begin{Example}
Fig. \ref{fig: example3} presents an example to construct the instance $\mathbb{G}=\{\mathbb{V},\mathbb{E}\}$ of the Steiner tree problem 
 with the graph $G$ shown in Fig. \ref{fig: given}, where $S=\{0,1\}$, $M=\{2,3,4,5,6,7\}$. The output $\mathbb{G}$ is 
presented in Fig. \ref{fig: complete}. SOFDA first replicates $G$ in $\mathbb{G}$. Subsequently, it duplicates the sources $0$ and $1$ by creating nodes $\hat{0}$ and $\hat{1}$, and VMs $2,3,4,5,6,7$ by creating nodes
 $\hat{2},\hat{3},\hat{4},\hat{5},\hat{6},\hat{7}$ in $\mathbb{G}$.
 Then, the costs of edges $(\hat{s},\hat{0}$), $(\hat{s},\hat{1}$), $(\hat{2},2$), $(\hat{3},3$), $(\hat{4},4$), $(\hat{5},5$), $(\hat{6},6$), and $(\hat{7},7$) are all set to $0$. To derive the cost of the virtual edge $(\hat{1},\hat{6})$, SOFDA finds the walk from source 1 to VM $6$ in $G$ as follows. First, it constructs an instance of the $k$-stroll problem $\mathcal{G}$ shown in Fig. \ref{fig: k-Stroll Instance C=5}. 
Then, we obtain a walk $W'_\mathcal{G}=W_\mathcal{G}=(1,4,2,3,5,6)$ in $\mathcal{G}$. By combining the shortest paths with each path connecting two consecutive nodes in $W'_\mathcal{G}$, we find the desired walk $W_G=(1,2,4,2,3,5,6)$ in $G$. Thus, the cost of link $(\hat{1},\hat{6})$ is set to the cost of $W_G$, which is equal to $c(2)+c(4)+c(3)+c(5)+c(6)+c(1,2)+c(2,4)+c(4,2)+c(2,3)+c(3,5)+c(5,6)=21$.
\end{Example}

The following lemma first indicates that the cost of the constructed Steiner
tree in $\mathbb{G}$ is bounded by $\rho _{ST}\cdot 3c(\mathcal{F}^{OPT})$,
by showing that there is a feasible Steiner tree $\mathbb{T}=\{\mathbb{V_{T}}%
,\mathbb{E_{T}}\}$ in $\mathbb{G}$ with the cost bounded by $3c(\mathcal{F}%
^{OPT})$.

\begin{lemma}
A feasible Steiner Tree with the cost no greater than $3c(\mathcal{F}^{OPT})$ exists in $\mathbb{G}$.
\label{theo: exist feasible 1}
\end{lemma}
\begin{proof}
We first show that there is a $\mathbb{T}$-like graph, $\mathbb{T^{\prime }}=\{%
\mathbb{V_{T^{\prime }}},\mathbb{E_{T^{\prime }}}\}$, with a cost of at most $3c(\mathcal{F}^{OPT})$ in $\mathbb{G}$. Afterward, we extract the desired $\mathbb{T}$ from $\mathbb{T}'$. 
Let $D_v^{OPT}$ denote the set of the destinations in the service overlay tree rooted at source $v$ in $\mathcal{F}^{OPT}$. In addition, for the service overlay tree rooted at source $v$ in $\mathcal{F}^{OPT}$, let $m_v^{OPT}$ be the representative last VM chosen from all the VMs
running $f_{|\mathcal{C}|}$ on the paths from $v$ to the destinations in $D_v^{OPT}$. Moreover, let $T_v$ be the optimal Steiner tree rooted at $m_v^{OPT}$ that spans all destinations in $D_v^{OPT}$ in $G$. Then, let $\mathbb{V_{T^{\prime }}}$ consist of 1) $\hat{s}$, 2) the duplicate $\hat{v}$ (in $\mathbb{V_S}$) of each source $v$ in $\mathcal{F}^{OPT}$, 3) the duplicate 
$\hat{m_v^{OPT}}$ (in $\mathbb{V_M}$) of each $m_v^{OPT}$ in $\mathcal{F}^{OPT}$, 4) each $m_v^{OPT}$ in $\mathcal{F}^{OPT}$, and 5) all VMs and switches (including all destinations in $D$) in all optimal Steiner trees $T_v$ in $G$. Let $\mathbb{E_{T^{\prime }}}$ include the edges between 1) $\hat{s}$ and $\hat{v}$, 2) $\hat{v}$ and $\hat{m_v^{OPT}}$, 3) $m_v^{OPT}$ and $\hat{m_v^{OPT}}$ for each spanned source $v$ in $\mathcal{F}^{OPT}$, and 4) all links in all optimal Steiner trees $T_v$ in $G$ for each used source $v$ in $\mathcal{F}^{OPT}$.

Note that for each source $v$ in $\mathcal{F}%
^{OPT}$, the cost of the edge between $\hat{v}$ and $\hat{m_v^{OPT}}$ in $%
\mathbb{T^{\prime }}$ is bounded by twice of the cost of the shortest walk
that visits $|\mathcal{C}|$ VMs between $v$ and $m_v^{OPT}$ in $G$. Since
there is a walk between $v$ and $v^{OPT}$ in $\mathcal{F}^{OPT}$, the total
cost of the edges in $\mathbb{E_{T^{\prime }}} \cap \mathbb{E}_{\mathbb{S}%
\mathbb{M}}$ is bounded by $2c(\mathcal{F}^{OPT})$. In addition, the cost
of $T_v$ is restricted by the connection cost of the service overlay tree with
root $v$ in $\mathcal{F}^{OPT}$, because the latter one not only spans $m_v^{OPT}$ 
and the destinations but also spans the source $v$ and other VMs (running
 $f_1,f_2,... ,f_{|\mathcal{C}|}$). Thus, the total cost of every edge in $%
\mathbb{E_{T^{\prime }}} \cap E$ is bounded by $c(\mathcal{F}^{OPT})$. Since
the cost of each edge in $\mathbb{E_{T^{\prime }}} \cap \mathbb{E}_{\hat{s}%
\mathbb{S}}$ or $\mathbb{E_{T^{\prime }}} \cap \mathbb{E}_{\mathbb{M}M}$ is
0, the cost of $\mathbb{T^{\prime }}$ is bounded by $3c(\mathcal{F}^{OPT})$%
. Furthermore, there is a subgraph (more specifically, a tree) $\mathbb{T}$ of 
$\mathbb{T^{\prime }}$ that spans the virtual node
and all the destinations in $\mathbb{G}$. 
Hence, the cost of $\mathbb{T}$ is smaller than that of $%
\mathbb{T^{\prime }}$ and is bounded by $3c(\mathcal{F}^{OPT})$.
\end{proof}

\subsection{Cost-Bounded Service Overlay Forest}\label{subsec: service overlay tree}
After finding a Steiner tree $\mathbb{T}$ in $\mathbb{G}$ with a bounded cost of at most $3\rho_{ST} c(\mathcal{F}^{OPT})$ by the above $\rho_{ST}$-approximation algorithm, to limit the total cost of the service overlay forest, SOFDA will deploy each service chain with the corresponding virtual edge in $\mathbb{T}\cap\mathbb{E}_{\mathbb{S}\mathbb{M}}$ and the route traffic via the edges in $\mathbb{T}\cap E$. Specifically, 
SOFDA first 1) adds each corresponding walk of the spanned virtual edge one by one in $G$ and then 2) adds all VMs, switches, and links in $\mathbb{T}\cap G$ to $F$.

\begin{Example} \label{ex: steiner tree 1}
Fig. \ref{fig: example4} presents an example for the construction of the service overlay forest with $\mathcal{C}=(f_1,f_2,f_3,f_4,f_5)$ in Fig. \ref{fig: given} by SOFDA.
First, an instance $\mathbb{G}=\{\mathbb{V},\mathbb{E}\}$ of the Steiner Tree problem is constructed with the input parameters $G$, $S=\{0,1\}$, $M=\{2,3,4,5,6,7\}$, and $\mathcal{C}=(f_1,f_2,f_3,f_4,f_5)$, and a Steiner tree $\mathbb{T}$ in $\mathbb{G}$ using the $\rho_{ST}$-approximation algorithm 
in \cite{SteinerTreeBestRatio} is obtained, as shown in Fig. \ref{fig: Steiner tree 1}.
\end{Example}

Nevertheless, multiple walks in $G$ corresponding to the spanned virtual edges in $\mathbb{T}$ may overlap in a few VMs, and the solution in this case is infeasible if any overlapping VM in this case needs to perform different VNFs (see the definition of SOF in Section \ref{sec: SOF problem}). The situation is called \emph{VNF conflict} in this paper. In the following, we present an effective way to eliminate the conflict by tailoring the overlapping walks without increasing the cost. To address the VNF conflict, when a walk $W_{G}=(v_{1},v_{2},\cdots ,v_{n})$ in $G$ is added to the service overlay forest $F$, it is encouraged to augment $F$ with a modified walk $W=(u_{1},u_{2},\cdots ,u_{n})$ based on $W_{G}$. Note that a VM or switch is allowed to be passed without processing any VNF by simply forwarding the data. 
Moreover,
a VNF conflict happens when two walks compete for a clone to perform different VNFs. Fig. \ref{fig: two walks with conflict} presents an example of the VNF conflict, where $W_1$ and $W_2$ respectively desire to run $f_1$ and $f_4$ on VM 4. Suppose that a walk $W$ (between source $s$ and VM $v$) 
faces
the VNF conflict with another walk $W_{1}$ (between source $s_{1}$ and VM $v_{1}$) in $F$. We 
solve
the
conflict between $W$ and $W_{1}$ by changing the source of $W$ from $s$ to $s_{1}$ (attaching $W$ to $W_{1}$), or changing the source of $W_{1}$ from $s_{1}$ to $s$ (attaching $W_{1}$ to $W$) without adding new links, VMs, and switches to $F$ and without enabling new VMs in $F$ for VNFs.

\begin{Example}  \label{ex: steiner tree 2}
Following Example \ref{ex: steiner tree 1}, SOFDA finds walks $W_{G,1}=(1,2,4,2,3,5,6)$ and $W_{G,2}=(0,3,5,3,2,4,7)$ in $G$, where $f_1$, $f_2$, $f_3$, $f_4$, $f_5$ are installed at VMs $4$, $2$, $3$, $5$, $6$ on $W_{G,1}$, and also at VMs $3$, $5$, $2$, $4$, $7$ on $W_{G,2}$, respectively.
After $W_{G,1}$ is added to $F$, we have $F=\{W_1\}$, where $W_1$ consists of one clone for source $1$, two clones of VM $2$, and one clone for VMs $4$, $3$, $5$, $6$ due to $F=\emptyset$ in the beginning. As $W_{G,2}$ is added to $F$, SOFDA augments $F$ with $W_2$, where $W_2$ includes one clone for source $0$, one clone for VMs $3$, $5$, $2$, $4$ (on which $f_3$, $f_4$, $f_2$, and $f_1$ are already running on $W_1$), and two new clones for VMs $3$ and $7$, as shown in Fig. \ref{fig: two walks with conflict}. 
\end{Example}

\begin{figure}
	\center
	\subfigure[]{\includegraphics[width=1.5cm]{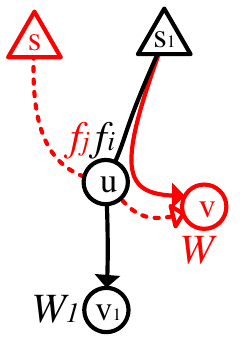}\label{fig: conflict 1}}\hfil
	\subfigure[]{\includegraphics[width=1.5cm]{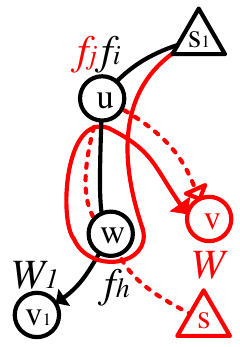}\label{fig: conflict 2}}\hfil
	\subfigure[]{\includegraphics[width=1.5cm]{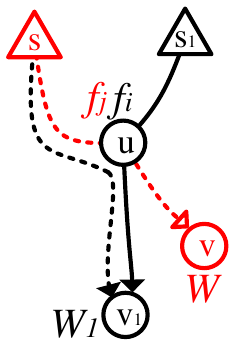}\label{fig: conflict 3}}
	\caption{Resolve of \emph{VNF conflicts} between two walks, where the black solid (or dashed) line denotes the original (or updated) $W_1$, and the red dashed (or solid) line denotes the original (or updated) $W_2$. Fig. \ref{fig: example5} (a), (b), and (c) show the resolve of the first, second, and third kinds of \emph{VNF conflicts}, respectively.} \label{fig: example5}
\end{figure}
Specifically, let $u$ be the first VM, where $W$ experiences the VNF conflict with $W_{1}$ by backtracking $W$. Recall in Fig. \ref{fig: two walks with conflict}, for example, that VM 4 is the first conflict node with $W_1$ by backtracking $W$. 
Let $f_{1},f_{2},\cdots ,f_{|\mathcal{C}|}$ denote the VNFs required to be performed in sequence on $W$ and $W_{1}$. Let $f_{i}$ and $f_{j}$ be the VNFs located at $u$ on $W_{1}$ and $W$, respectively. SOFDA effectively addresses the VNF conflict in details as follows.  

First, if $j\leq i$, SOFDA attaches $W$ to $W_{1}$ through $u$ by changing $W$ to the concatenation of the sub-walk of $W_{1}$ from $%
s_{1}$ to $u$ (on which $f_{1},f_{2},\cdots ,f_{i}$ are installed in sequence, identical to $W_{1}$) and the sub-walk of $W$ from $u$ to $v$ (on which $%
f_{i+1},f_{i+2},\cdots ,f_{|\mathcal{C}|}$ are running in sequence, identical to $W$), as shown in Fig. \ref{fig: conflict 1}.

\begin{Example} \label{ex: steiner tree 3}
Following Example \ref{ex: steiner tree 2}, $W_2$ first experiences the \emph{VNF conflict} with $W_1$ at (the clone of) VM $4$, where $f_4$ and $f_1$ are installed on $W_2$ and $W_1$, respectively. The sequence numbers of the VNFs at VM $4$ on $W_2$ and $W_1$ are $4$ and $1$, respectively. 
The condition $j\leq i$ is not satisfied since $j=4$ and $i=1$. SOFDA then checks the next condition. Note that one of the three conditions must be satisfied. 
\end{Example}

Second, if there is another VM $w$ such that $W$ experiences the \emph{VNF conflict} with 
$W_{1}$ at $w$, where $f_{h}$ with $h\geq j$ is on $W_{1}$, SOFDA attaches $W$ to $W_{1}$ through $w$ by changing $W$ to the concatenation of the sub-walk of $W_{1}$ from $s_{1}$ to $w$ (on which $f_{1},f_{2},\cdots ,f_{h}$ are running in sequence, identical to $W_{1}$), the sub-walk of $W$ from $w$ to $u$, and the sub-walk of $W$ from $u$ to $v$ (on which $f_{h+1},f_{h+2},\cdots,f_{|\mathcal{C}|}$ are running in sequence, identical to $W$), as illustrated in Fig. \ref{fig: conflict 2}.

\begin{Example}
Following Example \ref{ex: steiner tree 3}, $W_2$ experiences another \emph{VNF conflict} with $W_1$ at VM $5$, where $f_2$ and $f_4$ are performed on $W_2$ and $W_1$, respectively. Since the sequence number of the VNF at VM $5$ on $W_1$ is not smaller than that of the VNF at VM $4$ on $W_2$, SOFDA attaches $W_2$ to $W_1$ through VM $4$ as follows. SOFDA first steers $W_2$ along the sub-walk of $W_1$ from source $1$ to VM $5$ (i.e., the walk $(1,2,4,2,3,5)$) on which $f_1,f_2,f_3,f_4$ are running in sequence at VMs $4$, $2$, $3$, $5$, respectively, identical to $W_1$. 
Subsequently, it continues steering $W_2$ along the sub-walk of $W_2$ from VM $5$ to VM $4$ (i.e., the walk $(5,3,2,4)$), and the sub-walk of $W_2$ from VM $4$ to VM $7$ (i.e., the walk $(4,7)$) on which $f_5$ is run at VM $7$, identical to $W_2$. Finally, the sub-walk $(5,3,2,4,7)$ on the revised $W_2$ can be shortened to be a walk $(5,7)$. The constructed service overlay forest for $G$ is displayed in Fig. \ref{fig: overlap}.
\end{Example}

Otherwise, SOFDA attaches $W_{1}$ to $W$ through $u$ by changing $W_{1}$ to the concatenation of the sub-walk of $W$ from $s$ to $u$ (on which $f_{1},f_{2},\cdots ,f_{j}$ are running in sequence, identical to $W$) and the sub-walk of $W_{1}$ from $u$ to $v_{1}$ (on which $%
f_{j+1},f_{j+2},\cdots ,f_{|\mathcal{C}|}$ are run in sequence, identical to $W_{1}$), as shown in Fig. \ref{fig: conflict 3}. Moreover, when a walk $W$ experiences the \emph{VNF conflict} with multiple walks $W_{1},W_{2},\cdots ,W_{l}$ in $F$ in sequence by backtracking $W$, SOFDA
resolves the \emph{VNF conflict} between $W$ and $W_{1},W_{2},\cdots ,W_{l}$ one-by-one. The following theorem derives the approximation ratio for SOFDA.

\begin{figure}
	\centering
	\subfigure[]{\includegraphics[width=3.6cm]{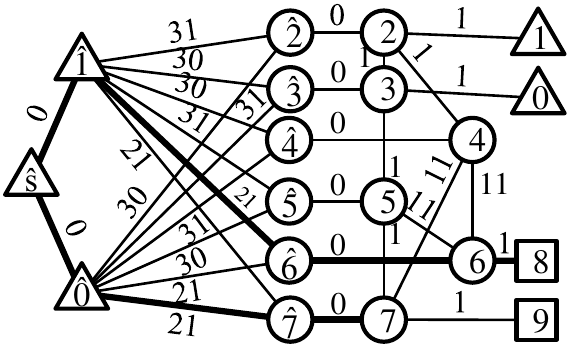}\label{fig: Steiner tree 1}}
	\hfill
	\begin{minipage}[b]{4 cm}
		\subfigure[]{\includegraphics[width=3.6cm]{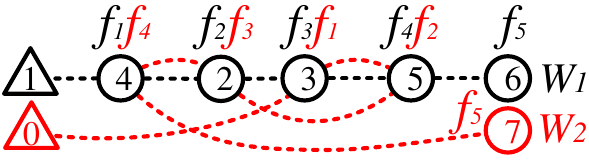}\label{fig: two walks with conflict}}
		\subfigure[]{\includegraphics[width=3.6cm]{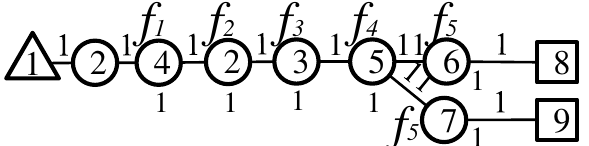}\label{fig: overlap}}
	\end{minipage}
	\caption{Example of construction of the service overlay forest by SOFDA. 
(a) The Steiner tree, shown in bold line. (b) Two walks with \emph{VNF conflict}. (c) The service overlay forest with $\mathcal{C}=(f_1,f_2,f_3,f_4,f_5)$ constructed for $G$ in Fig. \ref{fig: given}. }\label{fig: example4}
\end{figure}

\begin{theo}
	The cost of the constructed service overlay forest $F$ is bounded by $3\rho_{ST}c(\mathcal{F}^{OPT})$. \label{theo: exist feasible 2}
\end{theo}
\begin{proof}
First, the cost of Steiner tree $\mathbb{T}$ in $\mathbb{G}$ is bounded by $\rho_{ST}$ times of the optimal Steiner tree in $\mathbb{G}$. Since the cost of the optimal Steiner tree in $\mathbb{G}$ is bounded by $3c(\mathcal{F}^{OPT})$ according to Lemma \ref{theo: exist feasible 1}, the cost of $\mathbb{T}$ is limited by $3\rho_{ST}c(\mathcal{F}^{OPT})$. In addition, since the cost of the edge between $\hat{v} \in \mathbb{V_S}$ and $\hat{u} \in \mathbb{V_M}$ of $u$ in $%
\mathbb{T}$ is identical to the cost of the walk that visits $|\mathcal{C}|$ VMs between $v$ and $u$ in $G$, the cost of $F$ constructed in $G$ is equal to the cost of $\mathbb{T}$ and thereby bounded by $3\rho_{ST}c(\mathcal{F}^{OPT})$ if no \emph{VNF conflict} occurs in $F$. On the other hand, when the \emph{VNF conflict} between two walks happens, one of the two walks in $F$ is updated, and no new link, VM, and switch is added to $F$, and no VM in $F$ is newly created to perform the VNF. Thus, the cost of $F$ revised for resolving the \emph{VNF conflict} is still bounded by $3\rho_{ST}c(\mathcal{F}^{OPT})$. The theorem follows.
\end{proof}

\textbf{Time Complexity Analysis.}
We follow the notations in the time complexity analysis of SOFDA-SS. To generate the instance of the Steiner tree problem, SOFDA constructs $|S||M|$ instances of the $k$-stroll problem, and each of them employs the Dijkstra algorithm $|M|$ times to compute the edge costs of each instance. Then, SOFDA solves the $k$-stroll instance by \cite{k-MST03} to derive the costs of virtual edges (i.e., corresponding candidate service chains). To eliminate the conflict, in the worst case, all the added walks in $F$ are appended to the newly added walk, and the complexity is $O(|M|^3)$.
Therefore, the total time complexity is dominated by constructing and solving $k$-stroll instance and finding a Steiner tree, i.e., $O(|S||M|(|M|T_d+T_k)+T_s)$.

\section{Distributed Implementation} \label{sec: decentralize}

For large SDNs, it is important to employ multiple SDN controllers, where each one monitors and controls a subset of the network \cite{SDX,LocalAlgoInSDNs,applySDNtoTelecomHorizontalOrVertical},
and the communication protocols \cite{MultiControllerModel} between controllers are developed to facilitate scalable control of the whole network. In the following, therefore, we discuss the distributed implementation of the proposed algorithm in Section \ref{sec: general case} to support multi-controller SDNs. Note the controller that receives the request is elected to be the \emph{leader}, which is responsible for progress tracking and phase switching. 

First, shortest-path routing plays a fundamental role in SOFDA to build the auxiliary graph $\mathbb{G}$ and the service chain corresponding to each edge in $\mathbb{G}$. To find a shortest path traversing multiple domains, it is necessary for each controller to first abstract a matrix that consists of the lengths between every pair of border routers over the Southbound interface \cite{MultiControllerModel} within its domain. Afterward, each controller propagates the matrix to the other controllers along with the Network Layer Reachability Information of SDNi Wrapper over East-West Interface. 
which is used to share the connectivity information with the neighboring controllers. More specifically, let $s$ and $t$ denote the source and the destination, respectively. The controller $C_s$ covering $s$ can find the corresponding domain by the IP prefix of $t$. Then, controller $C_s$ informs the controller $C_t$ that covers $t$ of the lengths of all shortest paths from $s$ to all broader routers of $C_t$. Afterward, controller $C_t$ can respond the best broader router to controller $C_s$, and the length of a shortest path can be acquired accordingly.

Equipped with the shortest-path computation from multiple controllers, each controller can acquire the length of each shortest path between a VM in its domain and any other VM (or source). Thus, once the forest construction is initiated,
every controller that covers a source will communicate with other controllers to collect the matrices of lengths between any two VMs and the lengths between any source and any VM. Then, the controller can find all candidate service chains from its covered source to each VM and creates a virtual link in $\mathbb{G}$ representing the service chain to connect the virtual source and the corresponding last VM.

Afterward, a distributed Steiner tree algorithm \cite{PDSteinerTreeAlgo} can be employed by multiple controllers to find the Steiner tree, where the computation load originally assigned to each switch in the distributed algorithm can be finished by its controller instead. In SOFDA, it is important to address the VNF conflicts in multiple domains. To achieve this goal, each controller first removes the \emph{useless} candidate service chains that do not connect with any destination, and then informs any other controller whose coverage is visited by any remaining service chain. When one of the informed controllers observes a VNF conflict of two service chains, it notifies the other controller to collaboratively remove the conflict according to the conflict elimination algorithm described in Section \ref{subsec: service overlay tree}. Finally, each controller deletes the virtual source, deploys the remaining service chains, and forwards the content to the destinations by SOF.

\section{Discussion}\label{sec: discussion}
\subsection{Static Mulitcast Trees with Service Chaining} \label{subsec: static case}
To the best knowledge of the authors, this paper is the first one that explores the notion of the \textit{service forest}, i.e., the fundamental multi-tree multicast problem with service chaining, and provides approximation algorithms with theoretical bounds. Therefore, we first consider the fundamental problems for static SDN/NFV multicast and then extend the proposed algorithms to the dynamic case in Section \ref{subsec: dynamic adjustment}. 

Actually, static multicast is crucial for backbone ISPs. In this situation, each terminal node of a multicast tree is usually an edge router or a local proxy server of the ISP, instead of a dynamic user client. For example, current live streams are sent by the source (i.e., headends or content servers) and travel through the high-speed backbone network to the access nodes and edge nodes (e.g., Digital Subscriber Line Access Multiplexer (DSLAM) \cite{IPTV}, or a Mobile Edge Computing (MEC) server \cite{MEC-scenarios}) via \textit{static multicast} trees \cite{BenefitOfMulticast,IPTV,prejoinIPTV,surveyIPTV,measureCHT-MOD,MEC-scenarios} (e.g., Chung-Hwa Telecom MOD \cite{CHT-MOD}), whereas the dynamic user join and leave are handled by the local access nodes and edge nodes. Static multicast trees can significantly reduce the backbone bandwidth consumption for each stream and thereby is much more scalable to support a large number of video channels. In this case, each access node usually serves hundreds or thousands of end users and streams one (or few) channel(s) to each user according to the available bandwidth between the access node and user devices (e.g., set-top boxes). Moreover, for the massively multi-user virtual reality (e.g., gaming) \cite{MuVR,InterconnectedVirtualReality,MultiuserVR1,MultiuserVR2,MultiuserVR3}, the servers create a virtual environment with a 3D model, player avatars, and scripts, and then transmit the data by static multicast to several MEC servers \cite{MEC-scenarios}, which always need to appear in a multicast group. In the above life examples, our proposed algorithms can facilitate static multicast with service chaining (i.e., multiple stages of servers) to support a large number of streams between the headend server and the local access nodes/edge nodes.

\subsection{Cost Model and Online Deployment} \label{subsec: online deployment}
In the online scenario, when a new request arrives, SOFDA allocates the required resources for the request by constructing a service forest according to the current link and node costs. To balance the network resource consumption and accommodate more requests in the future, congested links and nodes are unnecessary to be assigned with higher costs for encouraging SOFDA to employ the links and nodes with low loads \cite{LinkWeightSetting,LinkWeightSetting2,NodeCost}. In this paper, therefore, we exploit \cite{LinkWeightSetting}, which is designed for online adaptive routing in the Internet, to assign a convex cost to each link or node. The cost will significantly increase as the load linearly grows, to avoid overwhelming the link or node. More specifically, let $l$ and $p$ denote current load and capacity of the link or node, respectively, and the cost $c$ is set according to the utilization (i.e., $l/p$) as follows and illustrated in Fig. \ref{fig: cost model}.
\begin{equation}\notag
c=
\begin{cases}
\begin{array}{ll}
l & \text{if } l/p \leq 1/3,\\
3l-2/3p & \text{else if } l/p \leq 2/3,\\
10l-16/3p & \text{else if } l/p \leq 9/10,\\
70l-178/3p & \text{else if } l/p \leq 1,\\
500l - 1468/3p & \text{else if } l/p \leq 11/10,\\
5000l-14318/3p  & \text{otherwise}.
\end{array}
\end{cases}
\end{equation}
\begin{figure}[t]
\center
\includegraphics[width=4.4cm]{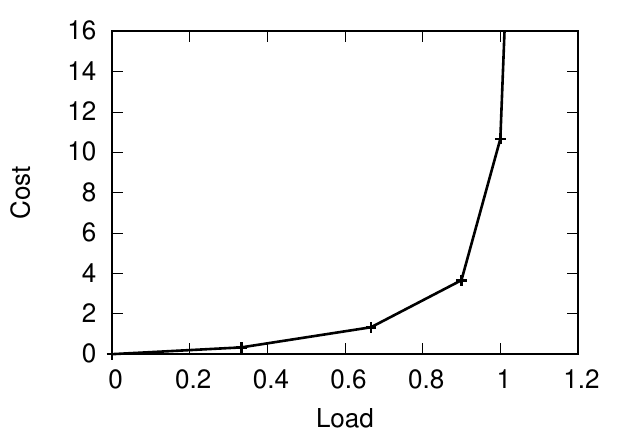}
\caption{The cost function with different load $l$ and capacity $p=1$.}
\label{fig: cost model}
\end{figure}

The cost model properly handles the online situation by assigning a huge cost to a more congested node or link. Therefore, SOFDA will avoid choosing the above congested node or link to minimize the total cost of the service forest. SOFDA thereby can mitigate the impact on a VM of other VMs colocated with an overloaded node. Indeed, the cost model can be applied to both private and public cloud networks, where resource optimization and load balancing are usually addressed. For example, Chung-Hwa Telecom MOD \cite{CHT-MOD} is built in its private clouds while Netflix \cite{Netflix} adopts AWS \cite{AWS}.

Nevertheless, each request has a different duration, and an approach without considering the duration of the request is inclined to incur fragmentation of the network resources and degrade the performance. However, the duration of a stream (e.g., a VR multi-player game) is usually difficult to be precisely predicted, and many current approaches thereby adaptively reroute \cite{MulticastReroute1,MulticastReroute2,Reroute3,Reroute4,Reroute5} and migrate the VM \cite{NodeCost,VMmigration,VMmigration2,VMmigration3,VMmigration4} to relocate the network resources when congestion occurs. Similar to the above approaches, when a node or link becomes congested, SOFDA reroutes the service forest by letting the users downstream to the above node or link re-join the forest again (explained in the reply of the first question), 
where the current path in the forest is removed only after the new join path is created to avoid service interruption \cite{MulticastReroute1,MulticastReroute2,Reroute3,Reroute4,Reroute5}.

\subsection{Adjustments for Various Dynamic Cases} \label{subsec: dynamic adjustment}.
In the following, we extend SOFDA to support the dynamic join and leave of destination users and the addition and deletion of NFVs in a service forest after a session starts. To address the dynamic case, a simple approach is to run SOFDA again for the whole forest. Nevertheless, this approach tends to incur massive computation loads in the SDN controller, especially when users frequently join and leave the multicast group or change the computation tasks in the service forest. In the following, therefore, we extend SOFDA to properly handle the dynamic case \cite{reliableMulticastforSDN,multicastTEforSDN}.
\begin{enumerate}
\item \textbf{Destination Leave.} When a destination $v$ leaves the service forest, if $v$ is a leaf node, SOFDA removes $v$ and all intermediate nodes and links in the path connecting $v$ and the closest upstream branch node in the service forest, where a branch node is a node in the forest with at least two child nodes. By contrast, if $v$ is not a leaf node, because there are other destination users in the subtree rooted at $v$,  SOFDA is not allowed to remove the path connecting to the upstream branch node.
\item \textbf{Destination Join.} When a new destination user $v$ joins the service forest, SOFDA finds the walk from $v$ to the forest with the lowest cost. More specifically, for each node $u$ in the forest $\mathcal{F}$ that can be a candidate branch node to connect $v$, let $f(u)$ denote the index of the last installed VNF between a source $s$ and $u$ in the forest. To derive the cost in the walk from $u$ to $v$, SOFDA finds the walk with $k=|\mathcal{C}|-f(u)+1$ from $u$ to $v$ to install the $(|\mathcal{C}|-f(u))$ new VNFs in the walk by exploiting $k$-stroll in the transformed graph (see Section \ref{sec: one source}). Let $W_G(u,v)=(u_1, ..., u_k)$ denote the acquired walk, where $u_1=u$ and $u_k=v$. In this case, SOFDA needs to install the new VNFs $f_{f(u)+1},...,f_{|\mathcal{C}|}$ on the above walk from $u$ to $v$, and the cost of the forest is increased by  $\min_{u\in\mathcal{F}}\{ c(W_G(u,v))\}$. SOFDA carefully examines every possible $u$ in the existing service forest to effectively minimize the increasing cost, and the node $u$ leading to the smallest cost is selected to serve the new destination user $v$ accordingly. 

\item \textbf{VNF Deletion.} When VNF $f_j$ is removed from the service forest, for each VM $v$ that installs a VNF $f_j$, SOFDA connects the VM $u$ with the upstream VNF $f_{j-1}$ to the VM $w$ with the downstream VNF $f_{j+1}$ (along the minimum-cost path from $u$ to $w$ in the original $G$) in the forest, where the source (or destination) can be regarded as the VM with the upstream (or downstream) VNF $f_{j-1}$ (or $f_{j+1}$) if $f_j$ is the first (or last) VNF.

\item \textbf{VNF Insertion.} When VNF $f_j$ is inserted to the service forest, for each pair of VMs $u$ and $w$ with VNFs $f_{j-1}$ and $f_{j+1}$, respectively, SOFDA installs $f_j$ on an available VM $v$, and connects $u$ to $v$ and $v$ to $w$ in the forest such that the sum of 1) the connection cost of the path between $u$ and $v$, 2) the setup cost of $v$, and 3) the connection cost of the path between $v$ and $w$ is minimized. When $f_j$ is the first (or last) VNF, the source (or destination) is regarded as the VM with VNF $f_{j-1}$ (or $f_{j+1}$). In addition, if two pairs of VMs $(u_1, w_1)$ and $(u_2, w_2)$ with VNFs $f_{j-1}$ and $f_{j+1}$ choose the same VM $v$ to install VNF $f_j$, SOFDA removes all intermediate nodes and links in the path connecting $u_2$ and $v$ in the forest in order to reduce the total cost of the forest (i.e., avoid creating redundant paths in the forest).

\item \textbf{Link Congestion.} For any congested link $e$ between the VMs with VNF $f_j$ and $f_{j+1}$, SOFDA updates the link cost according to \cite{LinkWeightSetting} and then re-connects the two VMs with the path associated with the lowest cost. SOFDA can effectively avoid choosing a congested link because the cost of the link will be extremely large. On the other hand, if $e$ is between the source and a VM (or VM and a destination), the source (or the destinations) is regarded as the upstream VM (or the downstream VM) and handled in a similar manner.
\item \textbf{VM Overload.} For any overloaded VM $v$ between the VMs with VNF $f_{j-1}$ and $f_{j+1}$, SOFDA updates the node cost according to \cite{LinkWeightSetting} to find an available VM $v'$ and then re-connects it to the upstream VM and downstream VM with the path having the lowest cost. Therefore, SOFDA can also avoid selecting an overloaded VM. On the other hand, if $v$ is the first VNF (or the last VNF), the source (or the destinations) is regarded as the upstream VM (or the downstream VM) and then  handled in a similar manner.
\end{enumerate}

\section{Numeric Result} \label{sec: simulation}
\subsection{Simulation Setup}

We conduct simulations to compare different approaches in two
inter-data-center networks: IBM SoftLayer \cite{SoftLayer} and Cogent \cite%
{Cogent}. SoftLayer contains 27 access nodes with 49 links and 17 data
centers, whereas Cogent has 190 access nodes with 260 links and 40 data
centers. We also generate a synthetic network with 5000 access nodes, 10000 links
, and 2000 data centers by Inet \cite{Inet}. The edge costs and the node costs are set
according to \cite{LinkWeightSetting} and \cite{NodeCost} based on the corresponding loads (described in Section \ref{subsec: online deployment}), 
respectively.
The sources and destinations are chosen uniformly at random from the nodes
in the network. We examine the performance of different approaches in two
scenarios: one-time deployment and online deployment. Moreover, we also
implement all algorithms in a small-scale SDN with HP SDN
switches.

In the one-time deployment scenario, the link bandwidth is set to 100 Mbps, 
and each requested demand is set to 5 Mbps. The link usage is randomly chosen in $(0,1)$ so as to derive the edge cost according to \cite{LinkWeightSetting}. 
Also, the total number of VMs ranges in $\{5,15,25,35,45\}$, and each VM is
randomly attached to a data center.
The service chain length (i.e., the number of VNFs in the chain) ranges in $\{3,4,5,6,7\}$ in both Cogent and Softlayer networks. The number of destinations and candidate sources range in $\{2,4,6,8,10\}$ and $\{2,8,14,20,26\}$, respectively in both networks.
The default numbers of candidate sources, destinations, VMs, and service
chain length are 14, 6, 25, 3, respectively. 

Afterward, for the online deployment scenario, the node/link usages are zero
initially. Each data center has 5 VMs with the cost according to the host machine utilization. Afterward, we incrementally generate a new request, and the node costs and edge costs will be updated according to \cite{LinkWeightSetting}. The numbers of
destinations and candidate sources in the request are randomly chosen from
13 to 17 and 8 to 12 in Softlayer, and they are from 20 to 60 and from 10 to 30 in
Cogent, respectively. The number of demanded services in a request is 3.

We compare the proposed algorithm with the following ones. 1) CPLEX \cite%
{CPLEX}. It finds the optimal solution by solving the IP formulation
in Section \ref{subsec: IP}.
2) Enhanced Steiner Tree (eST). Since the Steiner tree
algorithm \cite{SteinerTreeBestRatio} does not select VMs in the tree, we
extend it for SOF as follows. We find the minimum-cost tree among all
Steiner trees rooted at different sources. Afterward, we construct the
shortest service chain that is closest to the tree from \cite%
{ServiceChaining,ServiceChainWithDeadline} and then connect it to the tree
with the minimum cost. 3) Enhanced algorithm for the NFV enabled multicast
problem (eNEMP). Since the algorithm for the NFV enabled multicast problem
(NEMP) \cite{NFV-multicast} does not support multiple sources and VNFs,
similar to the above extension, we construct a service chain and then
connect it to the tree, where the chain spans the VM that has been chosen in
the tree. Moreover, we enable eST and eNEMP to support multiple sources via the modification as follows. The idea is to iteratively add a service tree in the solution
until no tree can reduce the total cost. 
At each iteration, we elect the minimal-cost service tree among all candidate trees rooted at each unused source, run VNFs sequentially on unused VMs, and span all the destinations in $D$. 
To estimate the profit of tree addition, we calculate the total cost of
the current forest with the elected tree, where each destination is spanned and served by the closet tree.
Hence, we add the elected tree and proceed to the next iteration if it can decrease the total cost. 
Otherwise, we output the forest.
Furthermore, a special case with only one Steiner tree
connected with a service chain (denoted by ST in the figures) is also
evaluated.

\subsection{One-Time Deployment}

\begin{figure}[t]
\center
\subfigure[]{\includegraphics[width=4.4cm]{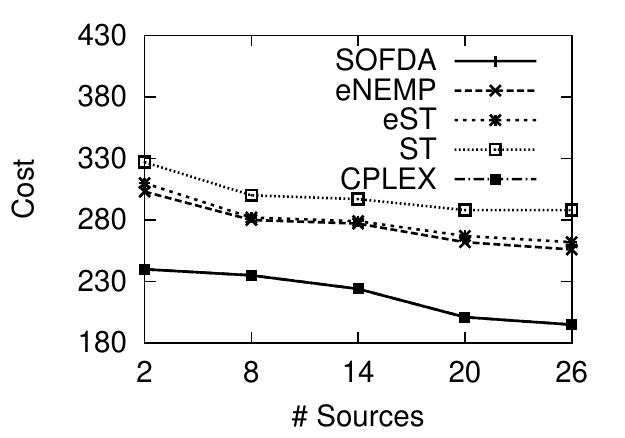}%
\label{fig: cost_num_src_small}}\hfil
\subfigure[]{\includegraphics[width=4.4cm]{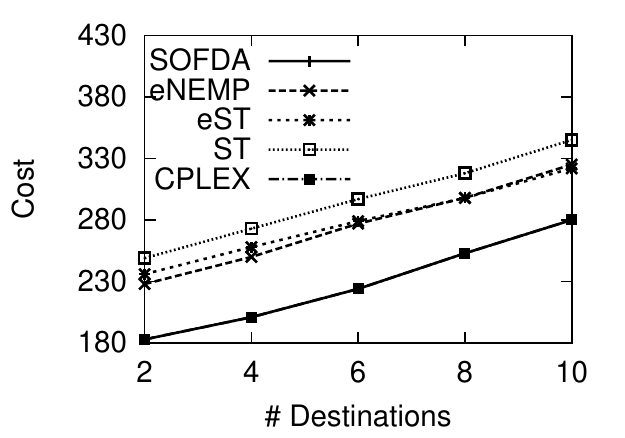}%
\label{fig: cost_num_dst_small}}\\[-10pt]
\subfigure[]{\includegraphics[width=4.4cm]{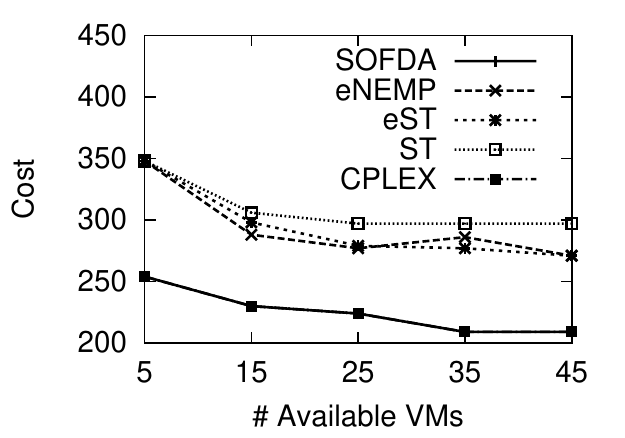}%
\label{fig: cost_num_vm_small}}\hfil
\subfigure[]{\includegraphics[width=4.4cm]{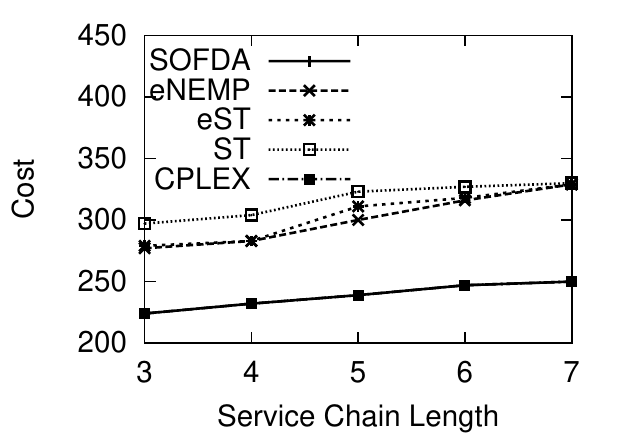}%
\label{fig: cost_num_vm_small}} 
\caption{The impact on cost of (a) the number of sources, (b) the number of
destinations, (c) the number of available VMs, and (d) the service chain
length in Softlayer network.}
\label{fig: simulation 1}
\end{figure}
\begin{figure}[t]
\center
\subfigure[]{\includegraphics[width=4.4cm]{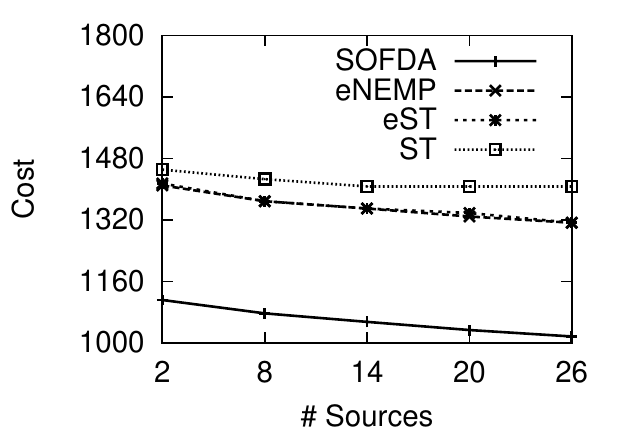}%
\label{fig: cost_num_src_large}}\hfil
\subfigure[]{\includegraphics[width=4.4cm]{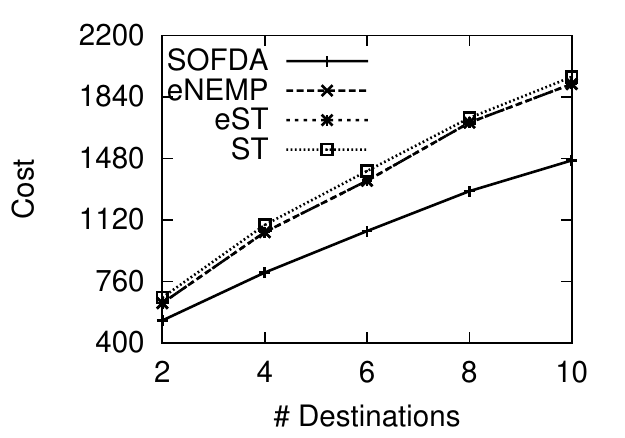}%
\label{fig: cost_num_dst_large}}\\[-10pt]
\subfigure[]{\includegraphics[width=4.4cm]{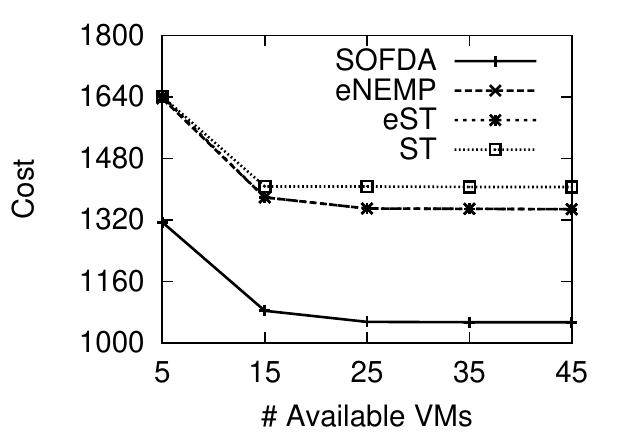}%
\label{fig: cost_num_vm_large}}\hfil
\subfigure[]{\includegraphics[width=4.4cm]{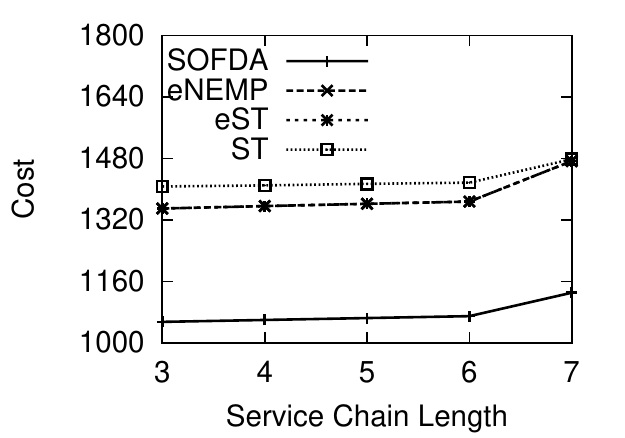}%
\label{fig: cost_num_vm_large}} 
\caption{The impact on cost of (a) the number of sources, (b) the number of
destinations, (c) the number of available VMs, and (d) the service chain
length in Cogent network.}
\label{fig: simulation 2}
\end{figure}

\begin{figure}[t]
\center
\subfigure[]{\includegraphics[width=4.4cm]{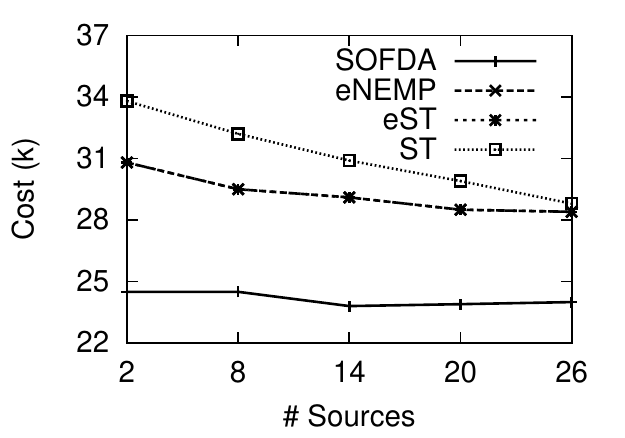}%
\label{fig: cost_num_src_large}}\hfil
\subfigure[]{\includegraphics[width=4.4cm]{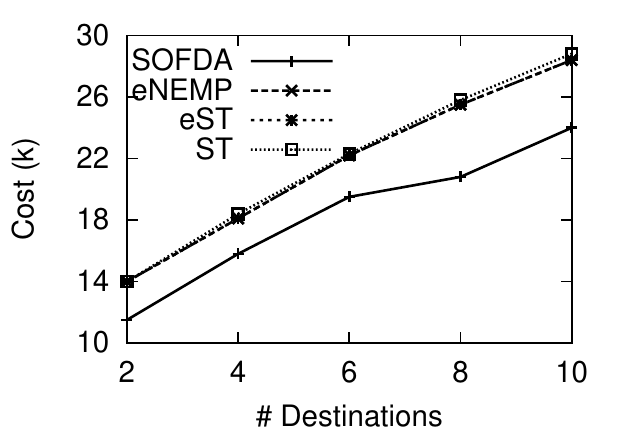}%
\label{fig: cost_num_dst_large}}\\[-10pt]
\subfigure[]{\includegraphics[width=4.4cm]{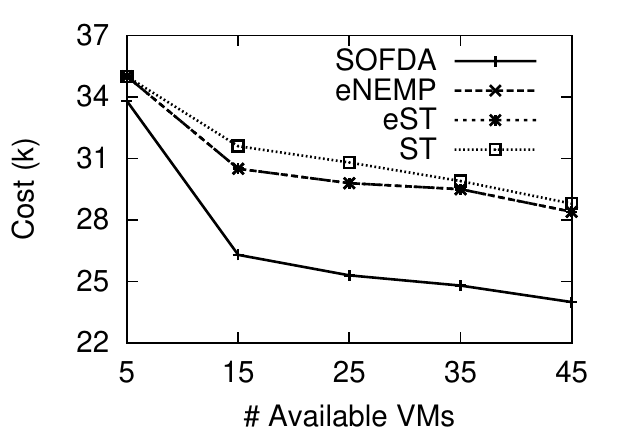}%
\label{fig: cost_num_vm_large}}\hfil
\subfigure[]{\includegraphics[width=4.4cm]{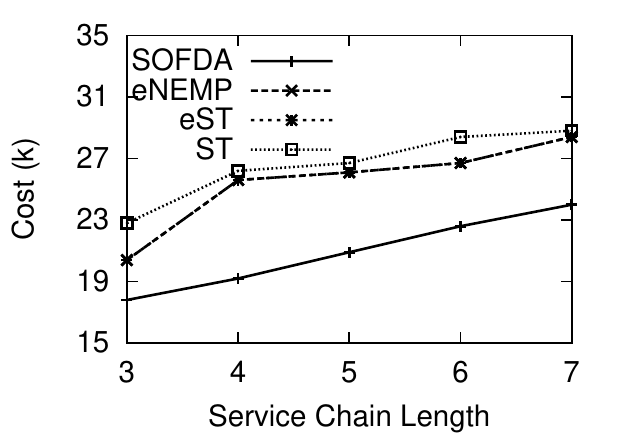}%
\label{fig: cost_num_vm_large}} 
\caption{The impact on cost of (a) the number of sources, (b) the number of
destinations, (c) the number of available VMs, and (d) the service chain
length in the synthetic network by Inet.}
\label{fig: simulation inet}
\end{figure}

We compare the performance of SOFDA, eNEMP, eST, ST, and the optimal solution generated by CPLEX with different numbers of sources, destinations, VMs, and different numbers of demand services. Because SOF is NP-hard, CPLEX is able to find the optimal solutions for small instances, and thus only Softlayer is tested in this case. Figs. \ref{fig: simulation 1} and \ref{fig: simulation 2} manifest that SOFDA is very close to the optimal solutions, and choosing multiple sources effectively reduces the total cost. The improvement in Fig. \ref{fig: simulation 2} and Fig. \ref{fig: simulation inet} is more significant because larger networks (i.e., Cogent and the synthetic network) contains more candidate nodes and links to generate a more proper forest. Since eNEMP and eST do not choose multiple sources and VMs during the multicast routing, they tend to miss many good opportunities for allocating the VMs with small costs to the tree with fewer edges. By contrast, the results indicate that SOFDA effectively reduces the total cost by 30\%. Also, when the number of sources increases, the destinations have more candidate trees to join, and thus the total cost is effectively reduced. However, the total cost increases when the number of destinations grows, because a service tree is necessary to span more destinations. Fortunately, when we have more VMs, there are more candidate machines to deploy VNFs, and the total cost thereby can be reduced.
Fig. \ref{fig: avg num of used VMs} presents the impact of different setup costs. The forest cost increases as the setup cost (i.e., 1x, 3x, ..., 9x) or the length of a demanded service chain (i.e., $|\mathcal{C}|$) grows as shown in Fig. \ref{fig: open cost 2}. Fig. \ref{fig: open cost 1} manifests that the average number of selected VMs in a forest is effectively reduced by SOFDA as the setup cost of a VM increases. Moreover, when the length of a demanded service chain (i.e., $|\mathcal{C}|$) becomes larger, the number of required VM needs to increase in order to satisfy new user requirements. 

Table \ref{table: speed_test} shows the running time of SOFDA with different numbers of sources and network sizes. The running time is less than 2 seconds for small networks, such as the one with 1000 nodes and 2 sources. With $|S|$ and $|V|$ increase, the running time grows, but SOFDA only requires around 19 seconds for the largest case.

\begin{figure}[t]
\center
\subfigure[]{\includegraphics[width=4.4cm]{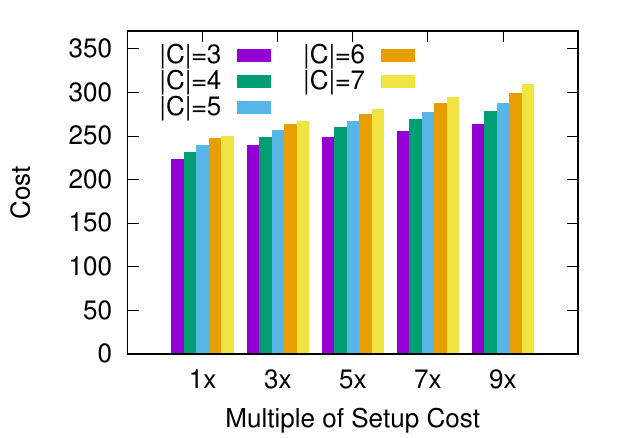}\label{fig: open cost 2}}\hfil
\subfigure[]{\includegraphics[width=4.4cm]{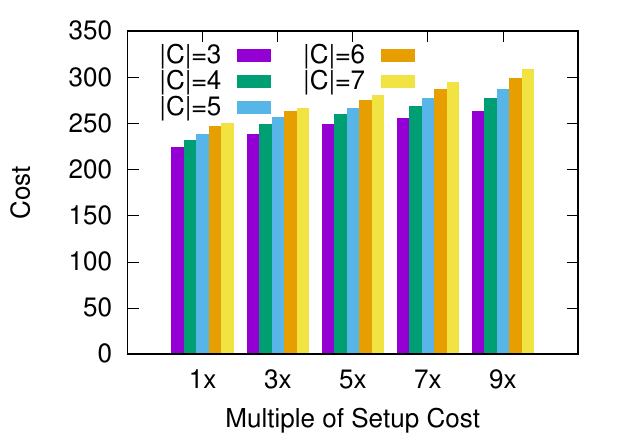}\label{fig: open cost 1}} 
\caption{The impact on (a) cost and (b) average number of used VMs by different multiples of setup cost and service chain length.}
\label{fig: avg num of used VMs}
\end{figure}

\subsection{Online Deployment}
\begin{figure}[t]
\center
\subfigure[]{\includegraphics[width=4.4cm]{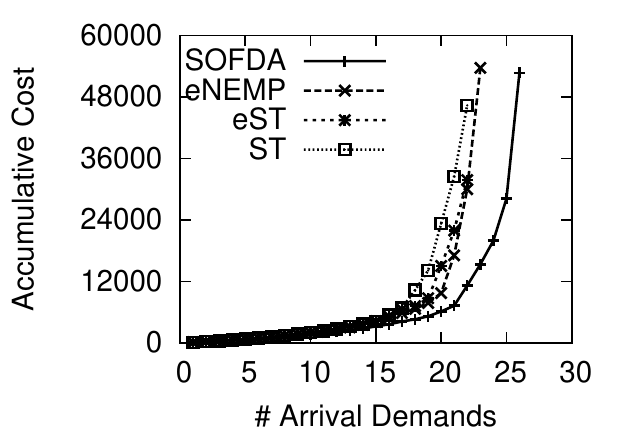}\label{fig: acc
cost 1}}\hfil
\subfigure[]{\includegraphics[width=4.4cm]{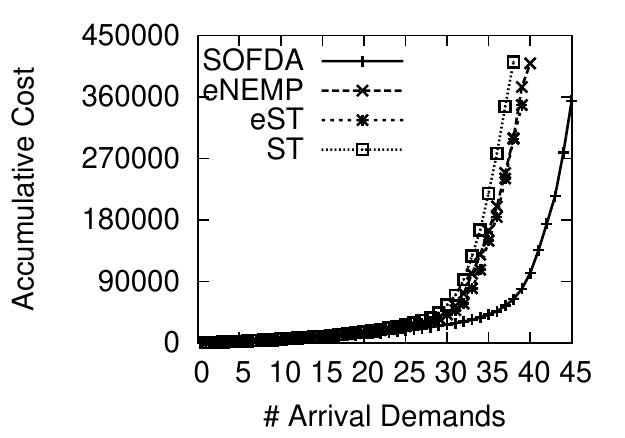}\label{fig: acc
cost 2}} 
\caption{Performance of the online deployment in (a) Softlayer network, and
(b) Congent network.}
\label{fig: simulation 3}
\end{figure}
In the following, we explore the online scenario with the requests
arriving sequentially. The edge costs also grow incrementally due to more
traffic demand. Fig. \ref{fig: simulation 3} presents the accumulative costs
(i.e., the total cost from the beginning to the current time slot) of
different approaches. It manifests that SOFDA outperforms the others because
the existing approaches focus on minimizing the traditional tree cost and
thus tend to miss many good opportunities to deploy the VNFs on a longer
path with sufficient VMs. By contrast, SOFDA carefully examines the edge
costs and node costs and acquires the best trade-off between utilizing more
VMs (leading to a smaller forest) and reducing the number of VMs, especially
when the network load increases.

\begin{table}[t]
\caption{The Running Time of SOFDA (Seconds)}
\label{table: speed_test}
\begin{center}
  \begin{tabular}{|c|c|c|c|c|c|} 
    \hline
    $|V|$ & $|S|=2$ & $|S|=8$ & $|S|=14$ & $|S|=20$ & $|S|=26$ \\ \hline
    1000  & 1.35    & 5.15    & 9.08     & 13.16    & 16.03    \\
    2000  & 1.48    & 5.622   & 9.76     & 13.77    & 17.172   \\
    3000  & 1.76    & 5.84    & 9.84     & 14.26    & 18.81    \\
    4000  & 1.89    & 6.15    & 9.88     & 14.86    & 18.9     \\
    5000  & 2.25    & 6.87    & 10.99    & 15.75    & 19.65    \\
    \hline
  \end{tabular}
\end{center}
\end{table}

\subsection{Implementation}
\begin{figure}[t]
\center
\subfigure[]{\includegraphics[width=3.5cm]{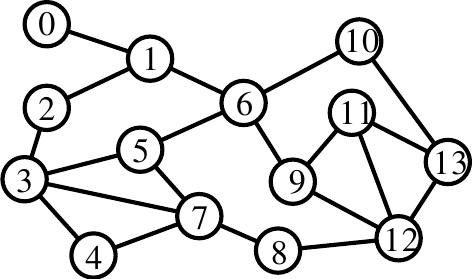}\label{fig:exp 5}}
\caption{Topology of Experimental SDN.}
\label{fig:exp 5}
\end{figure}
\begin{table}[t]
\caption{Evaluations in our experimental SDN network and Emulab}
\label{t:implementation}
\begin{center}
\begin{tabular}{|c|c|c|c|c|}
\hline
\multirow{2}{*}{Algorithms} & \multicolumn{2}{|c|}{Startup Latency} & \multicolumn{2}{|c|}{Re-buffering Time} \\
& \multicolumn{1}{c}{Ours} & \multicolumn{1}{c|}{Emulab} & \multicolumn{1}{c}{Ours} & \multicolumn{1}{c|}{Emulab}\\ \hline
SOFDA & 7.5 s & 5.5 s & 34.0 s & 29.8 s \\ 
eNEMP & 9.0 s & 5.9 s & 39.5 s & 39.0 s \\ 
eST & 10.0 s & 6.2 s & 41.0 s & 45.7 s \\ \hline
\end{tabular}%
\end{center}
\end{table}
To evaluate SOF in real environments, we implement SOFDA in Emulab \cite{Emulab}. The version and build of the Emulab are 4.570 and 03/17/2017, respectively. We create the topology by using NS format defined by Emulab and run Ubuntu 14.04 in each end host. We also deploy an experimental SDN with HP Procurve 5406zl OpenFlow-enabled switches and HP DL380G8 servers, where OpenDaylight is the OpenFlow controller, and OpenStack is employed to manage VMs. To support distributed computation, we run multiple OpenDaylight instances in VMs deployed in different servers and leverage the ODL-SDNi architecture \cite{MultiControllerModel}, which enables inter-controller communications. In addition, SOFDA is implemented as an application on the top of OpenDaylight and relies on OpenDaylight APIs to install forwarding rules into the switches. It also calls OpenStack APIs to launch VM instances, which are enabled VNFs.
The goal is to evaluate the transcoded and watermarked video performance under the environment with limited resources. Our testbed includes 14 nodes and 20 links, where the link capacity is set as 50 Mbps, and each node can support one VNF as explained in Section \ref{sec: SOF problem}. Two nodes are randomly selected as the video sources connecting to Youtube, and the full-HD test video is in 137 seconds encoded by H.264 with the average bit rate as 8 Mbps. Four nodes are randomly selected as destinations playing the videos with the VLC player. The video streams are processed by VNFs, including a video transcoder and a watermarker implemented by FFmpeg before reaching the destinations.
The available bandwidth of each link ranges from 4.5 Mbps to 9 Mbps to emulate the scenario with the network congestion, where the video playback may stall and wait for startup again or re-buffering. During the video playback, we measure the startup latency and total video re-buffering time, which are crucial for user QoE. Table \ref{t:implementation} summarizes the average startup latency and re-buffering time of different approaches. The results manifest that the startup latency of SOFDA is 20 \% and 33 \% shorter than eNEMP and eST, and the video stalling time of SOFDA is 16\% and 21\% smaller than eNEMP and eST, respectively. The experiment results also indicate that SOFDA routes traffic to less congested links compared with eNEMP and eST, and fewer packets thereby are lost. 

\section{Conclusion} \label{sec: conclusion}

In this paper, we investigated a new optimization problem (i.e., SOF) for cloud SDN. Compared with previous studies, the problem is more challenging because both the routing of a forest with multiple trees and the allocation of multiple VNFs in each tree are required to be considered. We proposed a $3\rho _{ST}$-approximation algorithm (SOFDA) to effectively handle the VNF conflict, which has not been explored by previous Steiner Tree algorithms. We also discussed the distributed implementation of SOFDA. Simulation results manifest that SOFDA outperforms the existing ones by over 25\%. Implementation results indicate that SOF can significantly improve the QoE of the Youtube traffic. Since current IP multicast supports dynamic group membership (i.e., each user can join and leave a tree at any time), our future work is to explore the online problem for rerouting of the forest and relocation of VNFs in cloud SDN with a performance guarantee.

\bibliographystyle{IEEEtran}
\bibliography{mybibfile}

\appendices





\section{The Proof of Theorem \ref{theo: NP-hard}} \label{sec: hardness}
\begin{proof}
We prove it by a polynomial-time reduction from a variant of the Steiner tree problem, where all the edge costs are positive and satisfy triangular inequality, to the SOF problem. Given any instance $G$ of the Steiner tree problem, we construct a corresponding instance $G'$ of the SOF problem as follows. We first replicate $G$ into $G'$, set $|\mathcal{C}|=1$ in $G'$, and add one source $s$ into $G'$. We let root $r$ as the only VM in $G'$ and nodes in $U$ of $G$ as the destinations in $D$ of $G'$. Root $r$ is connected to $s$ with an edge whose cost is set to an arbitrary value $w>0$ so as to obtain the instance $G'$. In the following, we prove $OPT_{G'}=OPT_G+w$. Because the edge $e_{r,s}$ only exists in $G'$ and the solution of $G'$ must contain a subgraph in $G$, which is also a tree rooted at $r$ and spans all the nodes in $U$, $OPT_{G}\leq OPT_{G'}-w$ holds. In addition, $OPT_{G}\geq OPT_{G'}-w$; otherwise, such a Steiner tree of $G$ plus the edge $e_{r,s}$ becomes a solution of $G'$ with a smaller cost than $OPT_{G'}$. Hence, given $OPT_G$ (or $OPT_{G'}$), we can obtain $OPT_{G'}$ (or $OPT_{G}$) by adding (or removing) the edge $e_{r,s}$. The theorem follows.
\end{proof}


\section{The Proof of Lemma \ref{lemma: triangular inequality}}\label{sec: triangular inequality}
\begin{proof}
Consider any three nodes $a$, $b$, and $c$ in $\mathcal{G}$. Since $\mathcal{G}$ is a complete graph, the edge between $a$ and $b$, the edge between $b$ and $c$, and the edge between $a$ and $c$ form a triangle in $\mathcal{G}$. Clearly, the total cost of the edge between $a$ and $b$ and the edge between $b$ and $c$ must be greater than the cost of the edge between $a$ and $c$; otherwise, the connection cost of the shortest path between $a$ and $c$ in $G$ must be greater than the total connection cost of the shortest path between $a$ and $b$ and the shortest path between $b$ and $c$, which is a contradiction. The lemma follows.
\end{proof}

\section{The Proof of Theorem \ref{theo: SOFDA-SS ratio}}\label{sec: SOFDA-SS ratio}
\begin{proof}
There are two possible cases for the last VM $u$. For the first case, $u$ is one of the last VMs in $\mathcal{F}^{OPT}$, and $u$ is not in the second case. For the first case, let
$\mathcal{F}^{OPT}(u)$ be a subgraph of $\mathcal{F}^{OPT}$ that connects $u$ to all the destinations in $D$. Thus, $c(\mathcal{F}^{OPT}(u))\leq c(\mathcal{F}_{E}^{OPT})$.
It is worthy to note that the subgraph $\mathcal{F}^{OPT}(u)$ may also span $s$ (and even the other last VMs in $\mathcal{F}^{OPT}$) so as to connect $u$ and all the destinations. Hence, the cost of the minimum Steiner tree that spans $u$ and all the destinations in $D$ must be no greater than $c(\mathcal{F}^{OPT}(u))\leq c(\mathcal{F}_{E}^{OPT})$ when $u$ also runs the last VNF in $\mathcal{F}^{OPT}$.
On the other hand, the $k$-stroll problem is NP-Hard and has a 2-approximation algorithm in metric graphs, which satisfy triangular inequality. According to Lemma \ref{lemma: triangular inequality}, $W_{\mathcal{G}}$ follows triangular inequality, and the cost of $W_{\mathcal{G}}$ is thereby no greater than twice of the cost of the shortest walk that visits at least $|\mathcal{%
C}|+1$ distinct nodes, $W_{\mathcal{G}}^{OPT}$, from $s$ to $u$ in $\mathcal{G}$. Since the cost of $W_{\mathcal{G}}^{OPT}$ is equal to the minimum setup and connection costs of the walk that visits at least $|\mathcal{C}|$ VMs from $s$ to $u$ in $G$, the cost of $W_{\mathcal{G}}^{OPT}$ is bounded by the total setup and connection costs of the walk from $s$ to $u$ in $%
\mathcal{F}^{OPT}$. Therefore, the cost of $W_{\mathcal{G}}$ is bounded by $2c(\mathcal{F}^{OPT})$. 


On the other hand, the connection cost of connecting $u$ to all destinations is bounded by $\rho _{ST}\cdot c(\mathcal{F}_{E}^{OPT})$. Thus, the cost of the service overlay forest with the last VM $u$ is bounded by $(2+\rho _{ST})c(\mathcal{F}^{OPT})$ as $u$ runs $f_{|\mathcal{C}|}$ in $\mathcal{F}^{OPT}$. Note that SOFDA-SS constructs a service overlay forest for every possible last VM $u$ and chooses the forest with the minimum cost. Since at least one VM runs $f_{|\mathcal{C}|}$ in $\mathcal{F}^{OPT}$, the cost of the service overlay forest generated by SOFDA-SS 
is bounded by $(2+\rho _{ST})c(\mathcal{F}^{OPT})$. The theorem follows.
\end{proof}

\clearpage
\section{Scenarios with Setup Costs on Sources} \label{sec: source cost}
In the following, we extend SOFDA-SS to support the case with a setup cost assigned to the source. Let $c(s)$ denote the cost to enable source $s$.
Let the cost of the edge between nodes $v_1$ and $v_2$ in $\mathcal{E}$
\begin{subequations}\notag
\begin{align}
c(v_1,v_2)=&\sum_{(a,b) \in P}c((a,b))\\
&+
\begin{cases}
\begin{array}{ll}
c(s)+c(u) & \text{if\ \ \ \ \ \ } v_1=s, v_2=u\\
\ & \text{\ \ \ \ \ or } v_1=u, v_2=s\\
\frac{c(s)+c(u)+c(v_2)}{2} & \text{else if } v_1=s, v_2\neq u\\
\ & \text{\ \ \ \ \ or } v_1=u, v_2\neq s\\
\frac{c(v_1)+c(s)+c(u)}{2} & \text{else if } v_1\neq s, v_2=u\\
\ & \text{\ \ \ \ \ or } v_1\neq u, v_2= s\\
\frac{c(v_1)+c(v_2)}{2} & \text{otherwise},
\end{array}
\end{cases}
\end{align}
\end{subequations}
where $u$ and $P$ denote the last VM and the shortest path between nodes $v_1$ and $v_2$ in $G$, respectively.
Due to a similar reason in Section \ref{sec: one source}, we set the cost of the edges in $\mathcal{E}$ in a similar way but consider the source cost $c(s)$.
The approximation ratio also holds, and the proof is similar to Theorem \ref{theo: SOFDA-SS ratio}.

\section{Pseudo Codes} \label{sec: pseudo code}
\begin{procedure} [h] \small
	\selectfont \caption{Instance Construction of the k-Stroll Problem} \label{Algo: k-stroll Instance Construction}
	\begin{algorithmic} [1]
		\REQUIRE A network $G=(V, E)$, a source $s \in V$, a set of VMs $M\subseteq V$, and a (last) VM $u \in M$\\
		\STATE $\mathcal{V} \leftarrow M \cup \{s\}$;\label{procedure: set V}
		\STATE $\mathcal{E} \leftarrow \{(x,y)|x,y \in \mathcal{V}\}$;\label{procedure: set E}
		\FOR {each edge $e$ in $\mathcal{E}$ between nodes $v_1$ and $v_2$}
			\STATE $P\leftarrow$ the shortest path between $v_1$ and $v_2$;
			\STATE \begin{equation}\notag c(v_1,v_2)\leftarrow\sum_{(a,b) \in P}c((a,b))+\begin{cases}\begin{array}{ll}
					\frac{c(u)+c(v_2)}{2} & \text{if } v_1=s,\\
					\frac{c(v_1)+c(u)}{2} & \text{else if } v_2=s,\\
					\frac{c(v_1)+c(v_2)}{2} & \text{otherwise};
					\end{array}\end{cases}
					\end{equation}
		\ENDFOR
		\STATE return $\mathcal{G}=\{\mathcal{V},\mathcal{E}\}$;
	\end{algorithmic}
\end{procedure}

\begin{procedure} [h] \small
	\selectfont \caption{Identification of the Walk with $|\mathcal{C}|$ VMs} \label{Algo: Walk Identification}
	\begin{algorithmic} [1]
		\REQUIRE A network $G=(V, E)$, a source $s \in V$, a set of VMs $M\subseteq V$, a (last) VM $u \in M$, and the length of the chain of VNFs $|\mathcal{C}|$\\
		\STATE construct an instance $\mathcal{G}=\{\mathcal{V},\mathcal{E}\}$ of the $k$-stroll problem using Procedure \ref{Algo: k-stroll Instance Construction} with input parameters $G$, $s$, $M$, and $u$; \label{procedure 1: construct k-stroll instance}\\
		\STATE obtain a walk $W_\mathcal{G}$ that visits at least $\mathcal{|C|}+1$ distinct nodes from $s$ to $u$ in $\mathcal{G}$ using the 2-approximation algorithm for the metric version of the $k$-stroll problem in \cite{k-MST03};\label{procedure 1: obtain W}\\
		\STATE obtain a walk $W'_\mathcal{G}=(u_1,u_2,\cdots, u_{\mathcal{|C|}+1})$, which visits exactly $\mathcal{|C|}+1$ distinct nodes from $s$ (=$u_1$) to $u$ (=$u_{\mathcal{|C|}+1}$), in $\mathcal{G}$ from $W_\mathcal{G}$ by repeatedly removing a node until $\mathcal{|C|}+1$ nodes remain;\label{procedure 1: obtain W'}\\
		\STATE obtain the walk $W_G=(u_1,u_2,\cdots, u_{\mathcal{|C|}+1})$ from $s$ (=$u_1$) to $u$ (=$u_{\mathcal{|C|}+1}$) by combining the shortest paths between $u_j$ and $u_{j+1}$ for $j=1, 2, \cdots, \mathcal{|C|}$ in $G$;\label{procedure 1: obtain W_G}\\
		\STATE return $W_G=(u_1,u_2,\cdots, u_{\mathcal{|C|}+1})$;
	\end{algorithmic}
\end{procedure}

\begin{algorithm} [h] \small
	\selectfont \caption{$(2+\rho_{ST})$-Approximation Algorithm for the SOF Problem with Single Source}
	\label{Algo: SOF_one_candidate_source}
	\begin{algorithmic} [1]
		\REQUIRE A network $G=\{V,E\}$, a set of destinations $D\subseteq V$, a source $s \in V$, a set of VMs $M\subseteq V$, and a chain of VNFs $\mathcal{C}=(f_1,f_2,\cdots, f_{|\mathcal{C}|})$\\
		\STATE $C \leftarrow \infty$;
		\FOR {each VM $u$} \label{algo: divide v adjust F}
		\STATE obtain the walk $W_G=(u_1,u_2,\cdots, u_{\mathcal{|C|}+1})$ from $s$ (=$u_1$) to $u$ (=$u_{\mathcal{|C|}+1}$) using Procedure \ref{Algo: Walk Identification} with input parameters $G$, $s$, $M$, $u$, and $\mathcal{|C|}$;\label{algo: obtain W_G}\\
		\STATE construct $F_{tmp}$ with cost $C_{tmp}$ by connecting $s$ with $u$ by $W_G$, running $f_1,f_2,\cdots, f_{|\mathcal{C}|}$ at $u_2,\cdots, u_{\mathcal{|C|}+1}$, respectively, and connecting $u$ with all destinations using the approximation algorithm for the Steiner Tree problem in \cite{SteinerTreeBestRatio};\label{algo: construct F_tmp}
		\IF {$C_{tmp} < C$} \label{algo: adjust F}
		\STATE $F \leftarrow F_{tmp}$;
		\STATE $C \leftarrow C_{tmp}$;
		\ENDIF \label{algo: adjust F end VMPDN2}
		\ENDFOR \label{algo: divide v adjust F end}
		\STATE return $F$; \label{algo: return F}
	\end{algorithmic}
\end{algorithm}

\begin{procedure} [h] \small
	\selectfont \caption{Instance Construction of the Steiner Tree Problem} \label{Algo: Steiner Tree Instance Construction}
	\begin{algorithmic} [1]
		\REQUIRE A network $G=(V, E)$, a set of sources $S\subseteq V$, a set of VMs $M\subseteq V$, and a chain of VNFs $\mathcal{C}=(f_1,f_2,\cdots, f_{|\mathcal{C}|})$\\
		\STATE $\mathbb{V}  \leftarrow V \cup \{\hat{s}\} \cup \mathbb{V_S} \cup \mathbb{V_M}$, where $\mathbb{V_S}$ is a set of the duplicate $\hat{v}$ of each $v \in S$, $\mathbb{V_M}$ is a set of the duplicate $\hat{v}$ of each $v \in M$, and $\hat{s}$ is a virtual source;\label{procedure 2: set V}
		\STATE $\mathbb{E}  \leftarrow E \cup \mathbb{E}_{\hat{s}\mathbb{S}} \cup \mathbb{E}_{\mathbb{S}\mathbb{M}} \cup \mathbb{E}_{\mathbb{M}M}$, where $\mathbb{E}_{\hat{s}\mathbb{S}}$ is a set of the edges $(\hat{s},\hat{v}$ for all $\hat{v} \in \mathbb{V_S}$, $\mathbb{E}_{\mathbb{S}\mathbb{M}}$ is a set of the edges $(\hat{v},\hat{u})$ for all $\hat{v} \in \mathbb{V_S}$ and $\hat{u} \in \mathbb{V_M}$, and $\mathbb{E}_{\mathbb{M}M}$ is a set of the edges $(v,\hat{v})$ for all $v \in M$;\label{procedure 2: set E}
		\FOR {each edge $e$ in $\mathbb{E}_{\hat{s}\mathbb{S}} \cup \mathbb{E}_{\mathbb{M}M}$} \label{procedure 2: weight of 0 begin}
		\STATE the weight of $e$ $ \leftarrow$ 0; \label{procedure 2: weight of 0}
		\ENDFOR \label{procedure 2: weight of 0 end}
		\FOR {each edge $(\hat{v},\hat{u})$, where $\hat{v} \in \mathbb{S}$ and $\hat{u} \in \mathbb{M}$, in $\mathbb{E}_{\mathbb{S}\mathbb{M}}$} \label{procedure 2: weight of not 0 begin}
		\STATE the weight of $(\hat{v},\hat{u})$ $\leftarrow$ the cost of the walk that visits $|C|$ VMs from $v \in S$ to $u \in M$ in $G$ using Procedure \ref{Algo: Walk Identification} with input parameters $G$, $v$, $M$, $u$, and $|C|$; \label{procedure 2: weight of not 0}
		\ENDFOR \label{procedure 2: weight of not 0 end}
		\STATE return $\mathbb{G}=\{\mathbb{V},\mathbb{E}\}$;
	\end{algorithmic}
\end{procedure}

\begin{algorithm} [h] \small
	\selectfont \caption{$3\rho_{ST}$-Approximation Algorithm for the SOF problem }
	\label{Algo: SOF}
	\begin{algorithmic} [1]
		\REQUIRE A network $G=\{V,E\}$, a set of destinations $D\subseteq V$, a set of sources $S\subseteq V$, a set of VMs $M\subseteq V$, and a chain of VNFs $\mathcal{C}=(f_1,f_2,\cdots, f_{|\mathcal{C}|})$\\
		\STATE construct an instance $\mathbb{G}=\{\mathbb{V},\mathbb{E}\}$ of the Steiner Tree problem using Procedure \ref{Algo: Steiner Tree Instance Construction} with input parameters $G$, $S$, $M$, and $\mathcal{C}$; \label{algo: construct Steiner tree instance}\\
		\STATE obtain a Steiner tree $\mathbb{T}$ in $\mathbb{G}$ using the $\rho_{ST}$-approximation algorithm for the Steiner Tree problem in \cite{SteinerTreeBestRatio}; \label{algo: construct Steiner tree}
		\STATE $F \leftarrow \emptyset$;\label{algo: initial F}
		\FOR {each pair of $v \in S$ and $u \in M$ in $G$ such that $\hat{v}$ (duplicate of $v$) is the parent of $\hat{u}$ (duplicate of $u$) in tree $\mathbb{T}$ with root $\hat{s}$} \label{algo: add walks to F begin}
		\STATE obtain the walk $W$ that visits $|C|$ VMs between $v$ and $u$ in $G$ using Procedure \ref{Algo: Walk Identification} with input parameters $G$, $v$, $M$, $u$, and $|C|$;\label{algo: obtain W}
		\STATE add $W$ to $F$ using Procedure \ref{Algo: F Augment} with input parameters $F$, $W$, and $\mathcal{C}$; \label{algo: add W to F}
		\ENDFOR \label{algo: algo: add walks to F end}
		\STATE add all VMs, switches, and links in $\mathbb{T} \cap G$ to $F$;
		\STATE return $F$; \label{algo: return F}
	\end{algorithmic}
\end{algorithm}

\begin{procedure} [h] \small
	\selectfont \caption{Augment of Service Overlay Forest} \label{Algo: F Augment}
	\begin{algorithmic} [1]
		\REQUIRE A service overlay forest $F$, a walk $W_G$ in $G$, and a chain of VNFs $\mathcal{C}=(f_1,f_2,\cdots, f_{|\mathcal{C}|})$\\
		\STATE obtain $W=(u_1,u_2,\cdots, u_n)$ from $W_G=(v_1,v_2,\cdots, v_n)$, where $u_i$:
		\begin{itemize}
		\item is the (clone of) VM $v_i$ running a VNF in $F$ if $v_i$ is a VM in $G$ and there is already (one clone of) VM $v_i$ running a VNF in $F$ and
		\item is the newly-added (clone of) VM (or a switch) $v_i$ in $F$ otherwise;
		\end{itemize}  \label{procedure 4: add W to F begin}
		\STATE augment $F$ with $W$;\label{procedure 4: add W to F end}
		\IF {$W$ experiences the \emph{VNF conflict} with at least one walk in $F$}
		\FOR {each walk $W_k$ that experiences the \emph{VNF conflict} with $W$ by backtracking $W$} \label{algo: resolve conflict begin}
		\IF {$W$ (bewteen source $s$ and VM $v$) experiences the \emph{VNF conflict} with $W_k$ (bewteen source $s_1$ and VM $v_1$) at VM $u$ for the first time and the sequence number of the VNF ($f_j$) on $W$ is not greater than that of the VNF ($f_i$) on $W_k$ at VM $u$}
		\label{algo: resolve conflict if 1 begin}
		\STATE attach $W$ to $W_k$ through $u$ by changing $W$ to the concatenation of the sub-walk of $W_k$ from $s_1$ to $u$ (on which $f_1,f_2,\cdots,f_i$ are run in sequence) and the sub-walk of $W$ from $u$ to $v$ (on which $f_{i+1},f_{i+2},\cdots,f_{|C|}$ are run in sequence); \label{algo: resolve conflict if 1}
		\STATE return $F$;
		\ELSIF {there exists another VM $w$ such that $W$ experiences the \emph{VNF conflict} with $W_k$ at $w$ and the sequence number of the VNF ($f_h$) at $w$ on $W_k$ is not smaller than that of the VNF ($f_j$) at $u$ on $W$}
		\STATE attach $W$ to $W_k$ through $w$ by changing $W$ to the concatenation of the sub-walk of $W_k$ from $s_1$ to $w$ (on which $f_1,f_2,\cdots,f_h$ are run in sequence), the sub-walk of $W$ from $w$ to $u$, and the sub-walk of $W$ from $u$ to $v$ (on which $f_{h+1},f_{h+2},\cdots,f_{|C|}$ are run in sequence); \label{algo: resolve conflict if 2-1}
		\STATE return $F$;\label{algo: resolve conflict if 2-2}
		\ELSE
		\STATE attach $W_k$ to $W$ through $u$ by changing $W_k$ to the concatenation of the sub-walk of $W$ from $s$ to $u$ (on which $f_1,f_2,\cdots,f_j$ are run in sequence) and the sub-walk of $W_k$ from $u$ to $v_1$ (on which $f_{j+1},f_{j+2},\cdots,f_{|C|}$ are run in sequence);
		\ENDIF \label{algo: resolve conflict if end}
		\ENDFOR \label{algo: resolve conflict end}
		\ENDIF
		\STATE return $F$;
	\end{algorithmic}
\end{procedure}

\begin{table}[h]
\small
\caption{Table of Notations}
\label{table: notation}
\begin{center}
  \begin{tabular}{ | c | l | } 
    \hline
    $G$ & the input network\\ \hline
    $V$ & the set of nodes in the network\\ \hline
    $M$ & the set of available VMs in the network\\ \hline
    $U$ & the set of switches in the network\\ \hline
    $S$ & the set of sources in the network\\ \hline
    $D$ & the set of destinations in the network\\ \hline
    $c(\cdot)$ & the cost function\\ \hline
    $\mathcal{C}$ & the demanded service chain\\ \hline
    $f_i$ & the $i^{th}$ VNF in chain $\mathcal{C}$\\ \hline
    $\mathcal{F}^{OPT}$ & the optimal service forest\\ \hline
    $\mathcal{F}^{OPT}_E$&  the set of used edges in the optimal forest\\ \hline
    $\mathcal{F}^{OPT}_M$& the set of used VMs in the optimal forest\\ \hline
    \multirow{2}{*}{$\mathcal{F}^{OPT}(u)$} & the subgraph of $\mathcal{F}^{OPT}$ which connects\\
    & node $u$ and all the destinations\\ \hline
    $D^{OPT}_v$ & the set of destinations whose source is $v$ in $\mathcal{F}^{OPT}$\\ \hline
    \multirow{2}{*}{$m^{OPT}_v$} & the representative last VM among all last VM\\ & on the walks from source $v$ in $\mathcal{F}^{OPT}$\\ \hline
    \multirow{2}{*}{{$T_v$}} & the Steiner tree that spans $m^{OPT}_v$ and the \\ & destinations in $D^{OPT}_v$ for source $v$ in $G$\\ \hline
    $\mathcal{G}$ & the constructed instance of $k$-stroll problem\\ \hline
    $\mathcal{V}$ & the nodes in $\mathcal{G}$\\ \hline
    $\mathcal{E}$ & the edges in $\mathcal{G}$\\ \hline
    $\mathcal{A}$ & the 2-approximation for $k$-stroll problem\\ \hline
    $W_\mathcal{G}$ & the output walk in $\mathcal{G}$ by $\mathcal{A}$\\ \hline
    \multirow{2}{*}{$W'_\mathcal{G}$} & the walk that visits exactly $|\mathcal{C}|$ distinct nodes \\ & in $\mathcal{G}$ obtained from $W_\mathcal{G}$\\ \hline
    $W_G$ & the walk in $G$ that visits $|\mathcal{C}|$ VMs in $G$\\ \hline
    $\mathbb{G}$ & the constructed instance of Steiner tree problem\\ \hline
    $\mathbb{V}$ & the set of nodes in $\mathbb{G}$\\ \hline
    $\mathbb{E}$ & the set of edges in $\mathbb{G}$\\ \hline
    $\hat{s}$ & the virtual source in $\mathbb{G}$\\ \hline
  \end{tabular}
\end{center}
\end{table}

\end{document}